%% file: paper.tex
\documentclass[sigconf]{acmart}

\usepackage{microtype}
\usepackage[utf8]{inputenc}

\usepackage{amsmath}
\usepackage{amssymb}
\usepackage{mathtools}
\usepackage{algorithm}
\usepackage{dsfont} 
\usepackage{flushend}

\newcommand{\E}{\textnormal{\textrm{E}}}

\DeclarePairedDelimiterX{\infdivx}[2]{}{}{ \left(#1\;\delimsize\middle\|\;#2\right) }

\newcommand{\dist}{\ensuremath{\textnormal{dist}}}
\newcommand{\simil}{\ensuremath{\textnormal{sim}}}

\newcommand{\LSH}{\mathcal{H}}
\newcommand{\ALSH}{\mathcal{A}}
\newcommand{\DLSH}{\mathcal{D}}
\newcommand{\DSH}{\mathcal{D}}
\newcommand{\ip}[2]{\langle{#1},{#2}\rangle}

\newcommand{\q}{\mathbf{q}}
\newcommand{\x}{\mathbf{x}}
\newcommand{\y}{\mathbf{y}}
\newcommand{\z}{\mathbf{z}}
\renewcommand{\a}{\mathbf{a}}
\newcommand{\vect}[1]{\mathbf{#1}}

\newcommand{\sphere}[1]{\mathbb{S}^{#1 - 1}}
\newcommand{\norm}[1]{\left\lVert #1 \right\rVert}
\newcommand{\cube}[1]{\{0,1\}^{#1}}
\newcommand{\pp}[1]{\mathcal{P}\left(#1\right)} 
\newcommand{\ppp}[1]{\mathcal{P}'\left(#1\right)}

\usepackage[textsize=tiny]{todonotes}

\copyrightyear{}
\acmYear{}
\setcopyright{none}
\acmConference{}{}{}
\acmBooktitle{}
\acmPrice{}
\acmDOI{}
\acmISBN{}
\begin{document}

\title{Distance-Sensitive Hashing}


\author{Martin Aumüller}
\affiliation{%
  \institution{BARC and IT University of Copenhagen}
}
\email{maau@itu.dk}

\author{Tobias Christiani}
\affiliation{%
  \institution{BARC and IT University of Copenhagen}
}
\email{tobc@itu.dk}

\author{Rasmus Pagh}
\affiliation{%
  \institution{BARC and IT University of Copenhagen}
}
\email{pagh@itu.dk}

\author{Francesco Silvestri}
\affiliation{%
  \institution{University of Padova}
}
\email{silvestri@dei.unipd.it}

\renewcommand{\shortauthors}{M. Aumüller et al.}

\begin{abstract}
\input{abstract}
\end{abstract}



\maketitle

\section{Introduction}\label{sec:introduction}
\input{intro}

\section{Optimal angular DSH}\label{angular}
\input{angular}

\section{Lower bound for monotone DSH}\label{sec:lower:bounds}
\input{lower_bounds}

\section{Hamming and Euclidean Space DSH}\label{sec:algos}
\input{algo}

\section{General constructions}\label{sec:generalconstr}
\input{constructions}

\section{Applications}\label{app:applications}
\input{applications}

\section{Conclusion}\label{sec:conclusion}
\input{conclusion}

\appendix 
\input{app_monotone_constructions}

\input{app_lower_sphere}

\input{app_euclidean_antilsh}
\input{app_general_constructions}

\bibliographystyle{ACM-Reference-Format}

\end{document}

%% file: abstract.tex
Locality-sensitive hashing (LSH) is an important tool for managing high-dimensional noisy or uncertain data, for example in connection with data cleaning (similarity join) and noise-robust search (similarity search).
However, for a number of problems the LSH framework is not known to yield good solutions, and instead ad hoc solutions have been designed for particular similarity and distance measures.
For example, this is true for output-sensitive similarity search/join, and for indexes supporting \emph{annulus queries} that aim to report a point close to a certain given distance from the query point.

In this paper we initiate the study of \emph{distance-sensitive hashing} (DSH), a generalization of LSH that seeks a family of hash functions 
such that the probability of two points having the same hash value is a given function of the distance between them.
More precisely, given a distance space $(X, \dist)$ and a ``collision probability function'' (CPF) $f\colon \mathbb{R}\rightarrow [0,1]$ we seek a distribution 
over pairs of functions $(h,g)$ such that for every pair of points $\x, \y \in X$ the collision probability is $\Pr[h(\x)=g(\y)] = f(\text{dist}(\x,\y))$.
Locality-sensitive hashing is the study of how fast a CPF can \emph{decrease} as the distance grows.
For many spaces, $f$ can be made exponentially decreasing even if we restrict attention to the symmetric case where $g=h$.
We show that the \emph{asymmetry} achieved by having a pair of functions makes it possible to achieve CPFs that are, for example, increasing or unimodal, and show how this leads to principled solutions to problems not addressed by the LSH framework.
This includes a novel application to privacy-preserving distance estimation.
We believe that the DSH framework will find further applications in high-dimensional data management.

To put the running time bounds of the proposed constructions into perspective, we show lower bounds for the performance of DSH constructions with increasing and decreasing CPFs under angular distance. Essentially, this shows that our constructions are tight up to lower order terms. In particular, we extend existing LSH lower bounds, showing that they also hold in the asymmetric setting.

%% file: intro.tex
The growth of data from a variety of sources that need to be managed and analyzed has made it increasingly important to design data management systems with features that make them robust and tolerant towards noisy data.
For example: different texts representing the same object (in data reconciliation), slightly different versions of a string (in plagiarism detection), or feature vectors whose similarity reflects the affinity of two objects (in recommender systems).
In data management, such tasks are often addressed using the similarity join operator~\cite{silva2010similarity}.

When data sets are \emph{high-dimensional}, traditional algorithmic approaches often fail.
Fortunately, there are general principles for handling high-dimensional data sets.
One of the most successful approaches is the locality-sensitive hashing (LSH) framework by Indyk and Motwani~\cite{IndykM98}, further developed in collaboration with Gionis~\cite{gionis1999similarity} and Har-Peled~\cite{Har-PeledIM12}.
LSH is a powerful framework for approximate nearest neighbor (ANN) search in high dimensions that achieves sublinear query time, and it has found many further applications.
However, for a number of problems the LSH framework is not known to yield good solutions, for example output-sensitive similarity search, and indexes supporting \emph{annulus queries} returning points at distance approximately~$r$ from a query point.

\smallskip

{\bf Motivating example.} A classical application of similarity search is in recommender systems: 
Suppose you have shown interest in a particular item, for example a news article $\x$.
The semantic meaning of a piece of text can be represented as a high-dimensional \emph{feature vector}, for example computed using latent semantic indexing~\cite{Deerwester90}.
In order to recommend other news articles we might search the set $P$ of article feature vectors for articles that are ``close'' to~$\x$.
But in general it is not clear that it is desirable to recommend the ``closest'' articles.
Indeed, it might be desirable to recommend articles that are on the same topic but are not \emph{too} aligned with $\x$, and may provide a different perspective~\cite{Abiteboul17}.

\smallskip

{\bf Discussion.} 
Unfortunately, existing LSH techniques do not allow us to search for points that are ``close, but not too close''.
In a nutshell: 
LSH provides a sequence of hash functions $h_1,h_2,\dots$ such that if $\x$ and $\y$ are close we have $h_i(\x)=h_i(\y)$ for some $i$ with constant probability, 
while if $\x$ and $\y$ are distant we have $h_i(\x) = h_i(\y)$ only with small probability (typically $1/n$, where $n$ upper bounds the number of distant points). 
As a special case, this paper discusses techniques that allow us to refine the first requirement:
If $\x$ and $\y$ are ``too close'' we would like collisions to occur only with very small probability.
At first sight this seems impossible because we will, by definition, have a collision when $\x = \y$. 
However, this objection is overcome by switching to an \emph{asymmetric} setting where we work with pairs of functions $(h_i,g_i)$ and are concerned with collisions of the form $h_i(\x)=g_i(\y)$.

More generally, we initiate the systematic study of the following question:
In the asymmetric setting, what is the class of functions $f$ for which it is possible to achieve $\Pr[h(\x)=g(\y)] = f(\text{dist}(\x,\y))$, 
where the probability is over the choice of $(h,g)$ and $\text{dist}(\x,\y)$ is the distance between $\x$ and $\y$. 
We refer to such a function as a \emph{collision probability function} (CPF).
More formally:
\begin{definition} \label{def:dsh} 
A \emph{distance-sensitive hashing} (DSH) scheme for the space $(X, \text{dist})$  is a distribution $\DLSH$ over pairs of functions $h, g \colon X \to \mathbb{R}$ with collision probability function (CPF) $f\colon \mathbb{R} \rightarrow [0,1]$ if for each pair 
 $\x,\y \in X$ and $(h, g)\sim\DLSH$ we have 
    $\Pr[h(\x) = g(\y)] = f(\text{dist}(\x,\y))$.
\label{def:distance:lsh}
\end{definition}
%
The theory of locality-sensitive hashing is the study of \emph{decreasing} CPFs whose collision probability is high for neighboring points and low for far-away points.

%


\subsection{Our contributions}\label{sec:results}

We initiate the systematic study of distance-sensitive hashing (DSH), and in particular we:
\begin{itemize}
\item Show tight upper and lower bounds on the maximum possible growth rate of CPFs under \emph{angular distance}. This extends upper and lower bound techniques for locality-sensitive hashing to the asymmetric setting.
\item Provide several DSH constructions that exploit asymmetry between the functions $g$ and~$h$ to achieve non-standard CPFs. 
For example, for Hamming distance we show how to achieve a CPF that equals any polynomial $\mathcal{P}\colon \{0,\dots,d\}\rightarrow [0,1]$ up to a scaling factor.
\item Present several motivating applications of DSH: Hyperplane queries, annulus search, spherical range reporting,  privacy-preserving distance estimation. 
\end{itemize}

The lower bound for angular distance implies lower bounds for Hamming and Euclidean distance DSH. 
It also shows that existing (asymmetric) LSH constructions used to search for vectors close to a given hyperplane~\cite{vijayanarasimhan2014hyperplane} are near-optimal.

On the upper bound side, our constructions show that asymmetric methods are significantly more expressive than standard, symmetric LSH constructions.
Since asymmetric methods are often applicable in cases where symmetric methods are used, it seems relevant to re-assess whether such constructions can be improved, even in settings where a decreasing CPF is desired.

Though our DSH applications do not lead to quantitative running time improvements compared to existing, published ad-hoc solutions, we believe that studying the possibilities and limitations of the DSH framework will help unifying approaches to solving ``distance sensitive'' algorithmic problems.
We now proceed with a more detailed description of our results.

 
%

\subsubsection{Angular distance}

We consider monotonically increasing CPFs for angular distance between vectors on the unit sphere $\sphere{d}$.
It will be convenient to express distances in terms of dot products $\ip{\cdot}{\cdot}$ rather than by angles or Euclidean distances, with the understanding that there is a 1-1 correspondence between them.

The constructions rely on the idea of ``negating the query point'': leveraging state-of-the-art symmetric LSH schemes, we obtain DSH constructions with monotone CPFs by replacing the query point $\q$ by $-\q$ in a symmetric LSH construction.
We initially apply this idea in section~\ref{sec:cross-polytope} to Cross-Polytope LSH \cite{andoni2015practical} getting an efficient DSH with a  monotonically decreasing CPF. 
We then show that a more flexible result follows with a variant of ideas used in the filter constructions from~\cite{BeckerDGL16,andoni2017optimal,christiani2017framework}. The filter based approach contains a parameter $t$ that can be used for fine tuning the scheme. This parameter is exploited in the data structure solving the annulus query problem (see section~\ref{app:annulus}). More specifically, the filter based approach gives the following result:
\begin{theorem}\label{thm:anti}
	For every $t > 1$ there exists a distance-sensitive family $\DSH_{-}$ for $(\sphere{d}, \ip{\cdot}{\cdot})$ 
	with a monotonically decreasing CPF $f$ such that for every $\alpha \in (-1,1)$ with $|\alpha| < 1 - 1/t$ we have
	\begin{equation}
		\ln(1/f(\alpha)) = \tfrac{1 + \alpha}{1 - \alpha}\tfrac{t^2}{2} + \Theta(\log t) .
	\end{equation}
The complexity of sampling, storing, and evaluating $(h,g) \in \DSH_{-}$ is $O(d t^4 e^{t^2 / 2})$.
\end{theorem}
Note that in terms of the angle between vectors $f$ is \emph{increasing}, as desired.
A corollary of Theorem~\ref{thm:anti} is that we can efficiently obtain a CPF $f$ such that $\rho_- = \ln f(0) / \ln f(\alpha) \leq \tfrac{1-\alpha}{1+\alpha} + o_t(1)$.
(We use $o_t(\cdot)$ to indicate that the function depends only on $t$.)
The value $\rho_-$ measures the gap between collision probabilities (the larger the gap, the smaller $\rho_-$), and is significant in search applications. The construction based on cross-polytope LSH has the same CPF as the construction stated in Theorem~\ref{thm:anti} with $t = \sqrt{2 \log d}$.

It turns out that the value $\rho_-$ is optimal up to the lower-order term.
To show this we consider vectors $\x,\y \in \{-1,+1\}^d$ that are either $0$-correlated or $\alpha$-correlated (i.e., $\Pr[\x_i = \y_i] = \tfrac{1+\alpha}{2}$ independently for each $i$).
These are unit vectors up to a scaling factor $\sqrt{d}$, and for large $d$ the dot product will be tightly concentrated around the correlation.
Since correlation is invariant under linear transformation, we may without loss of generality consider $\x,\y \in \{0,1\}^d$.
In section~\ref{sec:lower:bounds} we show the following lower bound:
\begin{theorem} \label{thm:revexphash} 
Let $\DSH$ be a distribution over pairs of functions $h, g \colon \cube{d} \to R$, 
and define $\hat{f} \colon [-1,1] \to [0,1]$ as $\hat{f}(\alpha) = \Pr[h(\x) = g(\y)]$ where $\x, \y$ are randomly $\alpha$-correlated and $(h,g) \sim \DSH$.  
Then for every $0 \leq \alpha < 1$ we have that $\hat{f}(\alpha) \geq \hat{f}(0)^{\frac{1+\alpha}{1-\alpha}}$.
\end{theorem}
Considering $f(\alpha) = \lim_{d\rightarrow\infty}\hat{f}(\alpha)$, taking logarithms of the final inequality we get that $\rho_- = \ln f(0) / \ln f(\alpha) \geq \tfrac{1-\alpha}{1+\alpha}$, as desired.
Theorem~\ref{thm:revexphash} extends standard (symmetric) LSH lower bounds~\cite{motwani2007,andoni2016tight,o2014optimal} to an asymmetric setting in the following sense: If there exists a too powerful asymmetric LSH family $\DSH_+$, then by negating the query point $\q$ we obtain a family $\DSH_-$ with a monotone CPF that contradicts the statement in Theorem~\ref{thm:revexphash}. Observe that this reasoning cannot be applied
to standard LSH bounds since they do not handle the asymmetry we allow in our setting.

The proof of the theorem builds on the Reverse Small-Set Expansion Theorem~\cite{odonnell2014analysis} which, given subsets $A, B \subseteq \cube{d}$ and random $\alpha$-correlated vectors $\x, \y$, lower bounds the probability of the event $[\x \in A \land \y \in B]$ as a function of $\alpha$, $|A|$ and $|B|$. 
Through a sequence of inequalities
we extend the lower bound for pairs of subsets of space to hold for distributions over pairs of functions that partition space, yielding a surprisingly powerful and simple lower bound.

\subsubsection{DSH constructions}
In the paper we also presents several specific and general constructions for DSH families, in addition to the constructions on the unit sphere discussed above.

In section~\ref{sec:hamming}, we apply the negating trick to 
the well-known bit-sampling approach from \cite{IndykM98} to get
a simple construction of a DSH family  with an increasing CPF for Hamming distance. The idea to obtain a DSH with increasing CPF is to use the bit-sampling approach, but negate the bit that is picked from the query point. We show that for small distances, the ``collision probability gap''  between points at distance $r$ and $r/c$, which is measured by the value $\rho_{-} = \ln f(r) / \ln
f(r/c)$, is worse than if we were to map
bit strings onto the unit sphere and use the construction from Theorem~\ref{thm:anti}.
This can be considered somewhat surprising because the bit-sampling LSH family has an optimal collision probability gap for the collision probability of points at distance $r$ and $cr$, measured by the value $\rho_+ = \ln f(r) / \ln f(cr)$, for small~$r$~\cite{o2014optimal}.

The negating trick does not work for Euclidean space, since it is potentially unbounded. Nevertheless, we describe in section~\ref{sec:antilshEucl} a DSH construction
with a $\rho_{-}$-value asymptotically matching the performance of constructions on the unit sphere.
The construction is based on an asymmetric version of the classical LSH family for Euclidean space by Datar et al.~\cite{Datar04}: specifically,
for parameters $w\in {\bf \mathbb{R}}$ and $k\in {\bf \mathbb{N}}$, we let
\begin{equation}\label{eq:E2LSH++}
h\colon \x \mapsto \left\lfloor \tfrac{\langle \a, \x\rangle + b
}{w}\right\rfloor,\qquad g\colon \x \mapsto \left\lfloor \tfrac{\langle \a, \x\rangle + b
}{w} \right\rfloor + k,
\end{equation} 
where $b \in [0, w]$ is uniformly random and $\mathbf{a} \sim \mathcal{N}^d(0,1)$ is a $d$-dimensional random Gaussian vector.
We show that this method, for a suitable choice of parameters $w$ and $k$, provides a near-optimal collision probability gap. This is surprising, since the classical construction in~\cite{Datar04} is not optimal as an LSH for Euclidean space~\cite{andoni2006}.
An example CPF for the above method is shown in Figure~\ref{fig:euclidean-cpf}. As we can see from this figure, the collision probaiblity decreases rapidly on the left side of the maximum, but decreases more slowly on the right side. In particular, this construction has a CPF that is not monotone but \emph{unimodal} (i.e., a distribution having a single local maximum).
\begin{figure}[t]
\includegraphics[width=.95\linewidth]{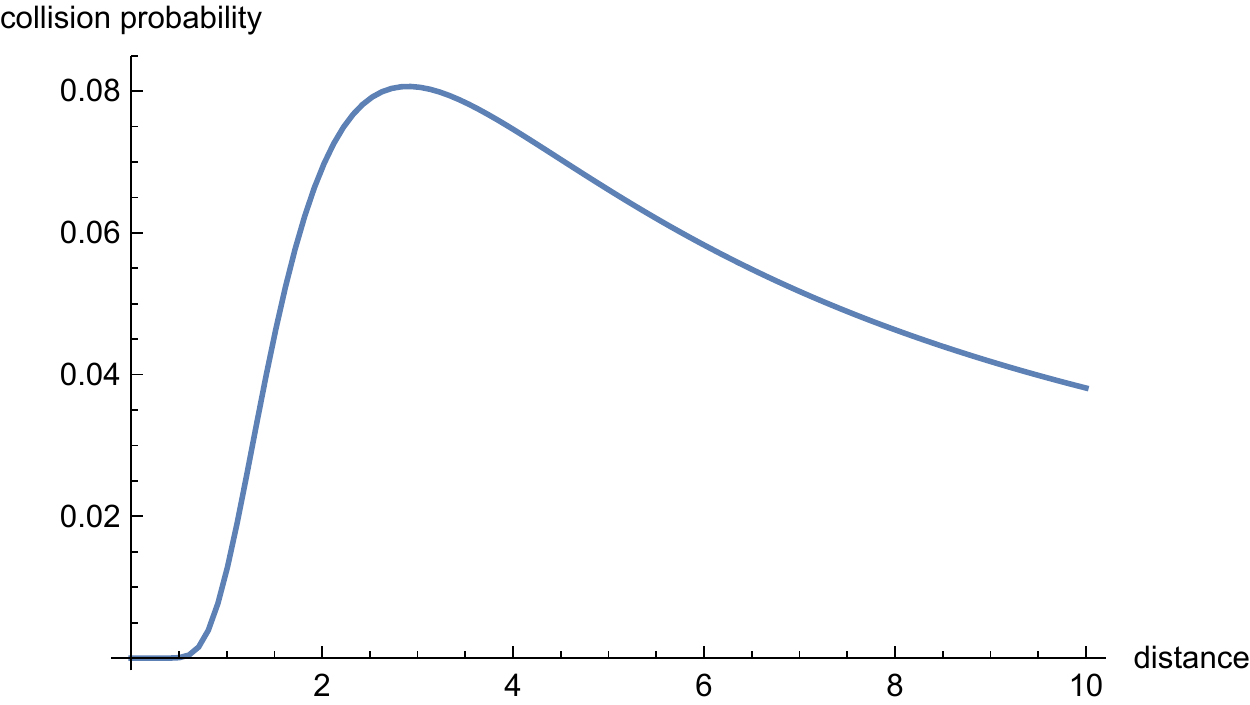}
\caption{Collision probability function of (\ref{eq:E2LSH++}) for $k=3$, $w=1$.}
\label{fig:euclidean-cpf}
\end{figure}

In section~\ref{sec:generalconstr}, we extend our work by targeting the following natural question: let $\pp{t}$ be a polynomial, does there exist 
a distance-sensitive hash family with CPF $f(t) = \pp{t}$?  We present two general approaches of constructing CPFs on the unit sphere and in Hamming space that cover a wide range of such polynomials. 
The result for the unit sphere easily follows from asymmetric embeddings previously used by Valiant~\cite{valiant2015finding} to solve the closest pair problem.
For  Hamming space, we propose an approach based on bit-sampling and polynomial factorization  that obtains the desired CPF up to  a scaling factor $\Delta\geq 1$ that depends on the roots of the polynomial. 

\subsubsection{Applications} 

We briefly describe here some applications of DSH constructions.
They are all straightforward given a DSH, but give an idea of the versatility of the framework.
We refer to section~\ref{app:applications} for formal statements and details.

\begin{itemize}
\item {\bf Hyperplane queries}. The problem of searching a set of unit vectors for a point (approximately) closest to a given hyperplane can be solved using a CPF $f$, parameterized by dot product, that peaks at $f(0)$. This approach was previously used in ad-hoc contructions, see~\cite{vijayanarasimhan2014hyperplane}.
\item {\bf Approximate annulus search}. The problem of searching for a point at approximately distance $r$ from a query point can be solved using a CPF that peaks at $r$. Previous solutions to this problem used a different, two-stage filtering approach~\cite{pagh2017approximate}.
\item {\bf Approximate spherical range reporting}. Classical LSH data structures are inefficient when many near neighbors need to be found, since each neighbor may have a high collision probability and may be ``found'' many times. A ``step function'' CPF that is flat for small distances, and then rapidly decreases implies good output-sensitivity. Again, previous results addressing this problem were ad-hoc~\cite{AhleAP17}.
\item {\bf Privacy preserving distance estimation.} If two parties want to compute the distance between private vectors, limiting the leakage of other information, secure multi-party computations can be used (see e.g.~\cite{goldreich2009foundations}).
Using ``step function'' CPFs we can transform this kind of question into a question about Hamming distance between vectors, for which much more efficient protocols exist~\cite{DBLP:conf/eurocrypt/FreedmanNP04,de2010linear}.
\end{itemize}

\subsection{Related work}\label{sec:related}

A recent book by Augsten and B{\"o}hlen~\cite{augsten2013similarity} surveys algorithms for similarity join, with emphasis on data cleaning.
Many of the commonly used algorithms are heuristics with weak theoretical guarantees, especially in high dimensions.
Recently, however, a substantial literature has been devoted the theoretical study of similarity search and join in data management, e.g.~\cite{ahle2016inner,DBLP:conf/pods/HuTY17,Kapralov15,DBLP:conf/kdd/Zhang017}.
All of these papers address particular similarity or distance measures, and their results are not directly comparable to those obtained in this paper.
In the following we review selected results from the LSH literature in more detail, referring to~\cite{wang2014lsh-survey,AndoniI16} for comprehensive surveys.

\smallskip

For simplicity we consider only LSH constructions that are \emph{isometric} in the sense that the probability of a hash collision depends only on the distance $\text{dist}(\x,\y)$.
In other words, there exists a CPF $f\colon \mathbb{R} \rightarrow [0,1]$ such that $\Pr[h(\x)=h(\y)]=f(\text{dist}(\x,\y))$.
Almost all LSH constructions whose collision probability has been rigorously analyzed are isometric. 
Notable exceptions are recent \emph{data dependent} LSH methods such as~\cite{andoni2015optimal} where the LSH distributions, and thus the collision probabilities, depend on the structure of data.

\smallskip

{\bf $\rho$-values.}
Much attention has been given to optimal so-called $\rho$-values of locality-sensitive hash functions, where we consider \emph{non-increasing} CPFs.
Suppose we are interested in hash collisions when $\text{dist}(\x,\y) \leq r_1$ but want to avoid hash collisions when $\text{dist}(\x,\y) \geq r_2$, for some $r_2>r_1$.
The $\rho$-value of this setting (denoted in this paper with $\rho_+$)  
 is the ratio of the logarithms of collision probabilities at distance $r_1$ and $r_2$, i.e., the real number in $[0,1]$ such that $f(r_1)=f(r_2)^\rho$.
The $\rho$-value determines the performance of LSH-based data structures for the \emph{$(r_1,r_2)$-approximate near neighbor} problem, see~\cite{Har-PeledIM12}.
In many spaces a good upper bound on $\rho$ can be given in terms of the ratio $c=r_2/r_1$, 
but in general the smallest possible $\rho$ can depend on $r_1$, $r_2$, $f(r_1)$, as well as the number of dimensions~$d$.
In our applications it will be natural to consider the ``dual'' $\rho$-value that measures the growth rate that can be achieved when distance increases.

\smallskip

{\bf LSHable functions.}
Charikar in~\cite{Charikar02} gave a necessary condition that all CPFs in the symmetric setting must fulfill, 
namely, $\dist(\x,\y) = 1 - \Pr[h(\x)=h(\y)]$ must be the distance measure of a metric, and more specifically this metric must be isometrically embeddable in~$\ell_1$.
In the asymmetric setting this condition no longer holds since in general $\Pr[h(\x)=g(\x)] < 1$.

Chierichetti et al. considered transformations that can be used to create new CPFs~\cite{Chierichetti15}, and studied the decision problem to verify if there exists an LSH with given pairwise collision probabilities~\cite{Chierichetti14}.
The transformations in~\cite[Lemma 7]{Chierichetti15} are considered in a symmetric setting, but the same constructions applied in the asymmetric setting give the following result, proved for completeness in Appendix~\ref{app:transform}:
\begin{lemma}
\label{lemma:transform}
Let $\{\mathcal{D}_i\}_{i = 1}^{n}$ be a collection of $n$ distance-sensitive families with CPFs $\{f_i\}_{i = 1}^{n}$. 
\begin{enumerate}
\item[\textnormal{(a)}] There exists a distance-sensitive family $\mathcal{D}_\text{concat}$ with CPF $f(x) = \prod_{i = 1}^n f_i(x)$. 
\item[\textnormal{(b)}] Given a probability distribution $\{p_i\}_{i = 1}^n$ over $\{\mathcal{D}_i\}$, there exists a distance-sensitive family $\mathcal{D}^p$ with CPF $f(x) = \sum_{i = 1}^n p_i f_i(x)$.
\end{enumerate}
\end{lemma}
An example application of Lemma~\ref{lemma:transform} is shown in Figure~\ref{fig:plateau}.
\begin{figure*}[t]
\hfill
\includegraphics[width=0.4\linewidth]{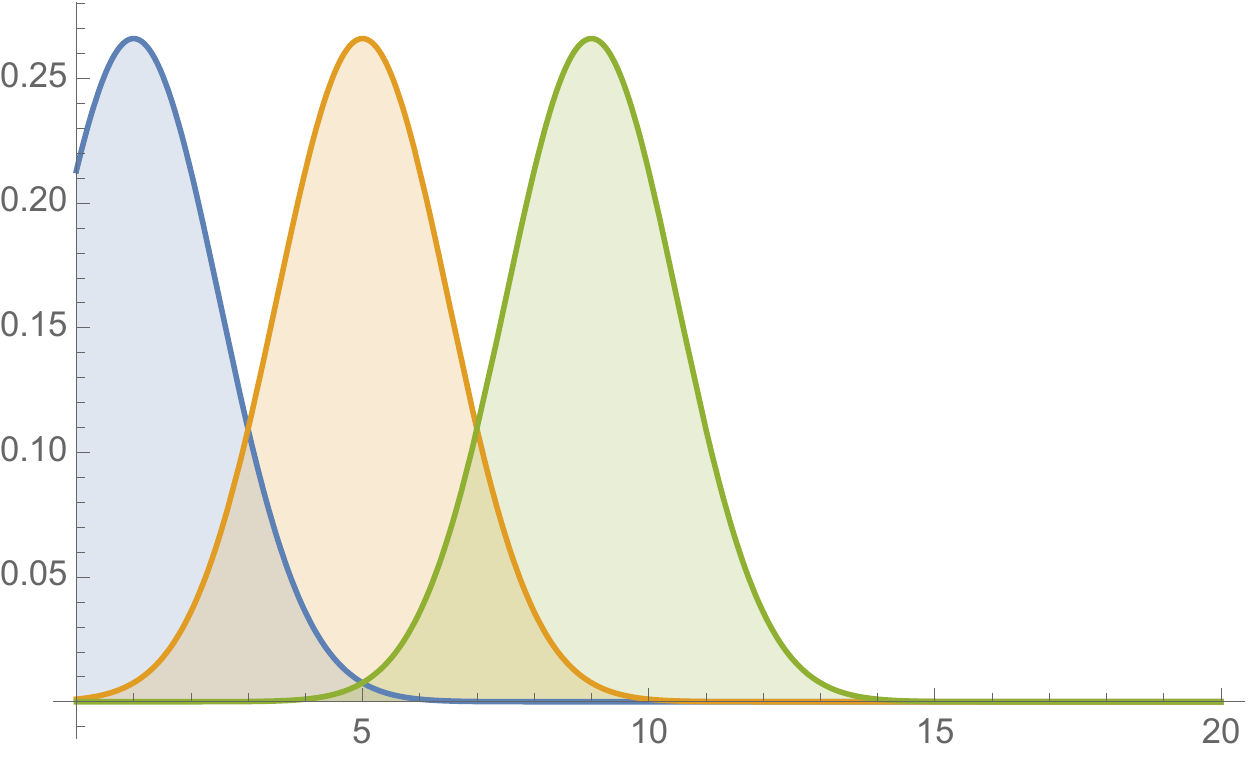}
\includegraphics[width=0.4\linewidth]{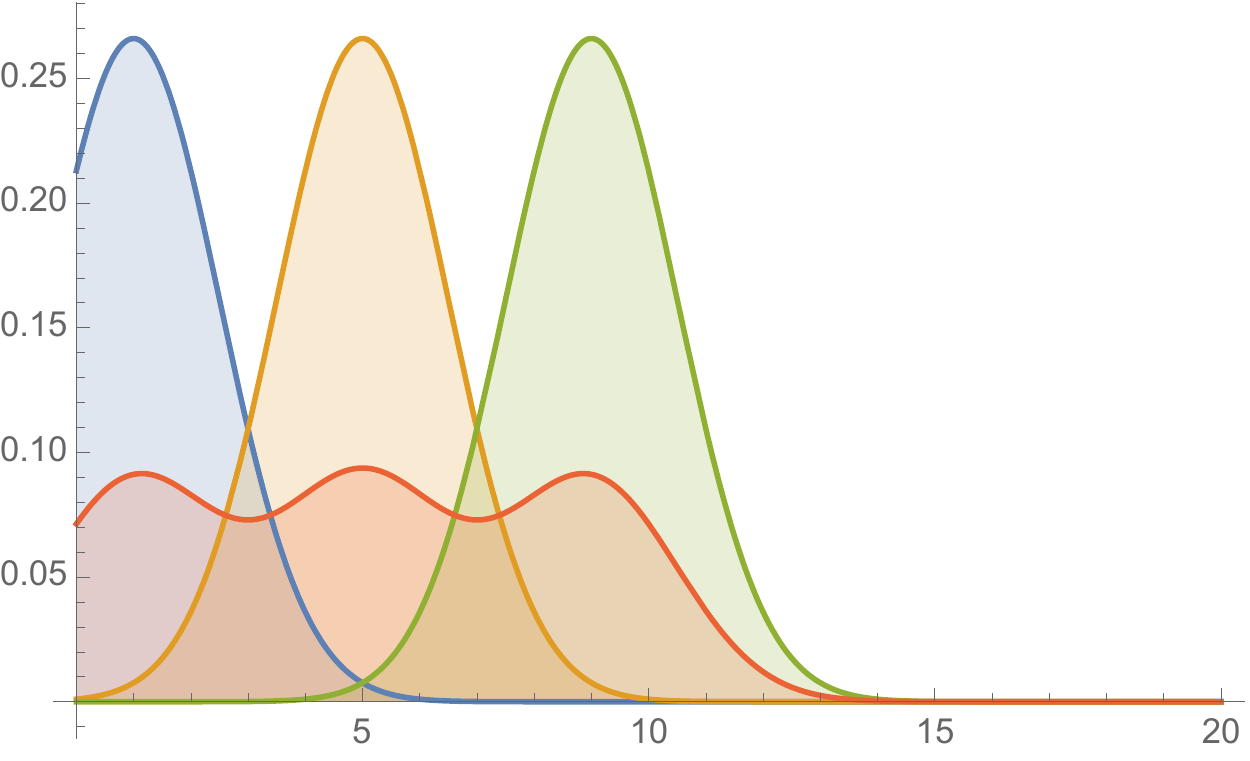}
\hfill\ 
\caption{Composing several unimodal CPFs (left) to form a ``step function'' CPF (red curve on the right) using Lemma~\ref{lemma:transform}. 
}
\label{fig:plateau}
\end{figure*}
%
%
Interestingly, at least in the symmetric setting, the application of this lemma to a single CPF yields \emph{all} transformations that are guaranteed to map a CPF to a CPF.
Chierichetti et al.~\cite{chierichetti2017distortion} recently extended the study of CPFs in the symmetric setting to allow \emph{approximation}, i.e., allowing the collision probability to differ from a target function by a given approximation factor.

\smallskip

{\bf Asymmetric LSH.} 
%
Motivated by applications in machine learning, Vijayanarasimhan et al.~\cite{vijayanarasimhan2014hyperplane} presented asymmetric LSH methods for Euclidean space where the CPF is a decreasing function of the dot product $|\langle \x, \y \rangle|$.
Shrivastava and Li~\cite{ShrivastavaL14} also explored how asymmetry can be used to achieve new CPFs (increasing), in settings where the inner product of vectors is used to measure closeness.
Neyshabur and Srebro~\cite{neyshabur15symmetric} extended this study by showing that the extra power obtained by asymmetry hinges on restrictions on the vector pairs for which we consider collisions:
If vectors are not restricted to a bounded region of $\mathbb{R}^d$, no nontrivial CPF (as a function of inner product) is possible.
On the other hand, if one vector is normalized (e.g.~a query vector), the performance of known asymmetric LSH schemes can be matched with a symmetric method.
But in the case where vectors are bounded but not normalized, asymmetric LSH is able to obtain CPFs that are impossible for symmetric LSH.
Ahle et al.~\cite{ahle2016inner} showed further impossibility results for asymmetric LSH applied to inner products, and that symmetric LSH is possible in a bounded domain even without normalization if we just allow collision probability~1 when vectors coincide.

Indyk~\cite{Indyk2003} showed how asymmetry can be used to enable new types of embeddings.
More recently asymmetry has been used in the context of locality-sensitive \emph{filters}~\cite{andoni2017optimal,christiani2017framework} and \emph{maps}~\cite{christiani2017beyond}.
The idea is to map each point $\x$ to a pair of sets $(h(\x),g(\x))$ such that $\Pr[h(\x)\cap g(\y) \ne\emptyset]$ is constant if $\x$ and $\y$ are close, and very small if $\x$ and $\y$ are far from each other.
This yields a nearest neighbor data structure that adds for each vector $\x \in P$ the elements of $h(\x)$ to a hash table; a query for a vector $\q$ proceeds by looking up each key in $g(\q)$ in the hash table.
One can transform such methods into asymmetric LSH methods by using min-wise hashing~\cite{Broder97,broder1997syntactic}, 
see~\cite[Theorem~1.4]{christiani2017framework}.

\smallskip

{\bf Recommender systems.}
Returning to our motivating example, the topic of getting ``interesting'' recommendations using nearest neighbor methods is not new.
Abbar et al.~\cite{AbbarAIMV13} built a nearest neighbor data structure on a \emph{core-set} of $P$ to guarantee diverse query results.
However, this method effectively discards much of the data set, so may not be suitable in all settings.
Indyk~\cite{Indyk2003} and 
Pagh et al.~\cite{pagh2017approximate} proposed data structures for finding the furthest neighbor in Euclidean space, leveraging random projections and using 
specialized data structures.

\smallskip

{\bf Privacy-preserving search.}
Privacy is an increasing concern in connection with data analytics.
Proximity information is potentially sensitive, since it may be used to reveal the source of a data point.
Ideally we would like information-theoretical privacy guarantees~\cite{Agrawal:2001:DQP:375551.375602}, but the standard technique of adding noise to data does not seem to work well for proximity problems, since adding noise merely shifts distances.
Riazi et al.~\cite{Riazi16} considered
answering nearest neighbor queries without leaking the actual distance.
They showed that standard LSH approaches can compromise privacy under a ``triangulation'' attack, however this risk can be reduced by designing an LSH (symmetric) with a CPF that is ``flat'' in the region of interest of an attacker.
However, only rather weak privacy guarantees were provided.

%% file: angular.tex

This section describes DSH schemes with monotonically increasing and decreasing CPFs for the unit sphere that match the lower bounds shown in the following section~\ref{sec:lower:bounds}.
As LSH are monotone schemes too (i.e., CPF decreasing with distance or increasing with similarity), we refer to DSH schemes with the opposite monotonicity as \emph{anti-LSH} (i.e., CPF increasing with distance or decreasing with similarity). 

For notational simplicity, the CPF is expressed as function of the inner product of two points, however it holds for other angular similarity and distance measures: on the unit sphere, inner product is equivalent to the cosine similarity and there is a 1-1 correspondence with Euclidean distance and angular distance. 
Results on the unit sphere can be extended to $\ell_s$-spaces for $0 < s \leq 2$ through Rahimi and Recht's~\cite{rahimi2007} embedding version of Bochner's Theorem~\cite{rudin1990} applied to the characteristic functions of $s$-stable distributions as used in~\cite{christiani2017framework}.  

We use the notation $\DSH_-$ to denote a family with CPF that is decreasing in the inner product between points, and $\DSH_+$ to denote a family with a CPF that is increasing in the inner product between points.
The idea behind the following constructions is to take a standard (symmetric) locality-sensitive hash family $\DSH_+$ for the unit sphere with an increasing CPF, and transform it into a family $\DSH_-$ with a decreasing CPF by introducing asymmetry (i.e, by negating the query value). 

\subsection{Cross-polytope DSH schemes}\label{sec:cross-polytope}

Andoni et al.~\cite{andoni2015practical} described the following LSH family $\mathcal{CP}$ for the $d$-dimensional unit sphere $\sphere{d}$: To sample a function $h \sim \mathcal{CP}$, sample 
a random matrix $A \in \mathbb{R}^{d \times d}$ in which each entry is independently sampled from a standard normal distribution $\mathcal{N}(0,1)$. To compute a hash value of a point $\x$, compute $\hat{\x} = A\x/\norm{A\x}\in \sphere{d}$ and map $\x$ to the closest point to $\hat{\x}$ among $\{\pm \mathbf{e_i}\}_{1 \leq i \leq d}$, where $\mathbf{e_i}$ is the $i$-th standard basis vector of $\mathbb{R}^d$. Intuitively, a random hash function from $\mathcal{CP}$ applies a random rotation to a point and hashes it to its closest point on the cross-polytope.

To formally define a DSH from $\mathcal{CP}$, we sample $(h_+, g_+) \sim \mathcal{CP}_+$ by sampling $h \sim \mathcal{CP}$ and set $h_+ = g_+ = h$. 
In \cite{andoni2015practical}, Andoni et al. showed the following theorem, here reproduced in terms of inner product similarity.\footnote{An inner product of $\alpha \in (-1,1)$ between two vectors on the unit sphere corresponds to Euclidean distance $\tau = \sqrt{2(1 - \alpha)}$.}
\begin{theorem}[Theorem 1 in \cite{andoni2015practical}] Let $f$ be the CPF of hash family $\mathcal{CP}_+$. Suppose that $\x,\y \in \sphere{d}$ such that $\ip{\x}{\y} = \alpha$, where $\alpha \in (-1,1)$. Then, 
\begin{align*}
\ln \frac{1}{f(\alpha)} = \frac{1 - \alpha}{1 + \alpha} \ln d + O_\alpha(\ln \ln d).
\end{align*}
\label{thm:crosspolytope}
\end{theorem}

To obtain a DSH with a monotonically decreasing CPF (in the similarity), we define the family $\mathcal{CP}_-$ consisting of pairs $(h_-, g_-)$ as follows: To sample a pair $(h_-,g_-) \sim \mathcal{CP}_-$, sample a function $h \sim \mathcal{CP}$. For each point $\x \in X$, set $h_-(\x) = h(\x)$ and $g_-(\x) = h(-\x)$. This means that we invert the query point before applying the hash function. Intuitively, we map the point to the point on the cross-polytope that is \emph{furthest away}
after applying the random rotation. 

\begin{corollary}Let $f$ be the CPF of hash family $\mathcal{CP}_-$. Suppose that $\x,\y \in \sphere{d}$ such that $\ip{\x}{\y} = \alpha$, where $\alpha \in (-1,1)$. Then, 
\begin{align*}
\ln \frac{1}{f(\alpha)} = \frac{1 + \alpha}{1 - \alpha} \ln d + O_\alpha(\ln \ln d).
\end{align*}
\end{corollary}

\begin{proof}
If $\ip{\x}{\y} = \alpha$, then $\ip{\x}{-\y} = -\alpha$. Since $h_-(\x) = g_-(\y)$ corresponds to $h(\x) = h(-\y)$, we can apply the result of Theorem~\ref{thm:crosspolytope} for similarity threshold $-\alpha$.  
\end{proof}

\subsection{Filter-based DSH schemes}\label{sec:unit_sphere}
Recent work on similarity search for the unit sphere \cite{andoni2014beyond, andoni2015optimal, BeckerDGL16, christiani2017framework} has used variations of the following technique: Pick a sequence of random spherical caps\footnote{A spherical cap is a portion of a sphere cut off by a plane.} and hash a point $\x \in \sphere{d}$ to the index of the first spherical cap in the sequence that contains $\x$. 
By allowing the spherical cap to have different sizes for queries and updates, it was shown how to obtain space-time tradeoffs for similarity search on the unit sphere that is optimal for random data~\cite{andoni2017optimal, christiani2017framework}.
We obtain a family $\DSH_{-}$ with a decreasing CPF by taking a standard (symmetric) family with its sequence of spherical caps 
and introduce asymmetry by negating the query point. Intuitively, this means that for $(h, g) \sim \DSH_{-}$ we let $h$ use the original sequence of spherical caps while $g$ uses the spherical caps that are \emph{diametrically opposite} to the ones used by $h$. 
We therefore get a collision $h(\x) = g(\y)$ if and only if $\x$ and $\y$ are contained in random diametrically opposite spherical caps.
We proceed by describing the family $\DSH_{+}$ and the modification that gives us the family $\DSH_{-}$.

The family $\DSH_{+}$ takes as parameter a real number $t > 0$ and an integer $m$ that we will later set as a function of $t$.
We sample a pair of functions $(h,g)$ from $\DSH_{+}$ by sampling $m$ vectors $z_{1}, \dots, z_{m}$ where $\z_{i} \sim \mathcal{N}^{d}(0,1)$. 
The functions $h,g$ map a point $\x \in \sphere{d}$ to the index $i$ of the first projection $\z_i$ where $\ip{\z_i}{\x} \geq t$.
If no such projection is found, then we ensure that $h(\x) \neq g(\x)$ by mapping them to different values. Formally, we set
\begin{align*}
	h_{+}(\x) &= \min(\{ i \mid \ip{\z_i}{\x} \geq t \} \cup \{ m + 1 \}), \\
	g_{+}(\x) &= \min(\{ i \mid \ip{\z_i}{\x} \geq t \} \cup \{ m + 2 \}). 
\end{align*}

We use the idea of negating the query point to obtain a family $\DSH_{-}$ from $\DSH_{+}$ by setting: 
\begin{equation*}
	g_{-}(\x) = g_{+}(-\x) = \min(\{ i \mid \ip{\z_i}{\x} \leq -t \} \cup \{ m + 2 \}). 
\end{equation*}

The analysis of CPF for $\DSH_{-}$ and $\DSH_{+}$ and the proof of Theorem~\ref{thm:anti} are provided in Appendix~\ref{app:tailbounds}.

%% file: lower_bounds.tex
This section provides  lower bounds on the CPFs of DSH families in $d$-dimensional Hamming space under the 
similarity measure $\simil_{H}(\x, \y) = 1 - 2 \norm{\x-\y}_{1}/d$.
These results extend to the unit sphere and Euclidean space through standard embeddings. 

Our primary focus is to obtain the lower bound in Theorem~\ref{thm:revexphash}, which holds for a CPF that is decreasing with the similarity.
As with our upper bounds for the unit sphere, re-applying the same techniques also yields a lower bound for the case of an increasing CPF in the similarity.
The proof combines the (reverse) small-set expansion theorem by O'Donnell~\cite{odonnell2014analysis} with techniques inspired by the LSH lower bound of Motwani et al.~\cite{motwani2007}. 
The main contribution here is to extend this lower bound for pairs of subsets of Hamming space to our object of interest: distributions over pairs of functions that partition space. 
We begin by introducing the required tools from~\cite{odonnell2014analysis}.
\begin{definition}\label{def:correlation}
	For $-1 \leq \alpha \leq 1$ and $\x, \y \in \cube{n}$ we say that $(\x, \y)$ is randomly $\alpha$-correlated 
	if $\x$ is uniformly distributed over $\cube{n}$ and each component of $\y$ is i.i.d.\ according to 
	\begin{equation*}
		\y_i =
		\begin{cases}
			\x_i & \text{with probability } \frac{1 + \alpha}{2}, \\
			1 - \x_i & \text{with probability } \frac{1 - \alpha}{2}. 
		\end{cases}
	\end{equation*}
\end{definition}
The reverse small-set expansion theorem lower bounds the probability that random $\alpha$-correlated points $(\x, \y)$ 
end up in a pair of subsets $A, B$ of the Hamming cube, as a function of the size of these subsets.  
In the following, for $A \subset \cube{d}$ we refer to the quantity $|A|/2^d$ as the \emph{volume} of $A$.
\begin{theorem}[Rev.\ Small-Set Expansion~\cite{odonnell2014analysis}]\label{thm:expansionreverse}
Let $0 \leq \alpha \leq 1$. Let $A, B \subseteq \cube{d}$ have volumes $\exp(-a^{2}/2)$, $\exp(-b^{2}/2)$, respectively, where $a, b \geq 0$. 
Then we have that
\begin{equation*}
	\Pr_{\substack{(\x, \y) \\ \alpha \text{-corr.}}} [\x \in A, \y \in B]\geq \exp\left(-\frac{1}{2}\frac{a^2 + 2\alpha a b + b^2}{1 - \alpha^2}\right).
\end{equation*}
\end{theorem}
We define a probabilistic version of the collision probability function that we will state results for. 
In Section \ref{sec:lower_concentration} we will apply concentration bounds on the similarity between $\alpha$-correlated pairs of points in order to make statements about the actual CPF.
We will use $R$ to denote the range of a family of functions which, without loss of generality, we can assume to be finite.
\begin{definition}[Probabilistic CPF]
	Let $\DSH$ be a distribution over pairs $h, g \colon \cube{d} \to R$. 
    We define the probabilistic CPF $\hat{f} \colon [-1,1] \to [0,1]$ by 
	\begin{equation*}
		\hat{f}(\alpha) = \Pr_{\substack{(h,g) \sim \DSH \\ \text{$(\x, \y)$ $\alpha$-corr.} }}[h(\x) = g(\y)].
	\end{equation*}
\end{definition}
The proof of the lower bound will make use of the following technical inequality that follows from two applications of Jensen's inequality.
\begin{lemma} \label{lem:inequality} 
	Let $p, q$ denote discrete probability distributions, then for every $c \geq 1$ we have that 
	\begin{equation*}
		\sum_{i} (p_i q_i)^c \geq \left(\sum_{i} p_i q_i \right)^{2c-1} 
	\end{equation*}
	with reverse inequality for $c \leq 1$.
\end{lemma}
\begin{proof}
	Assume $c \geq 1$.
	By Jensen's inequality, using the fact that $x \mapsto x^c$ and $x \mapsto x^{2 - 1/c}$ are convex we have that 
	\begin{equation*}
	\sum_{i} (p_i q_i)^c = \sum_{i} p_i(p_i^{1 - 1/c} q_i) 
						 \geq \left( \sum_i p_i^{2 - 1/c} q_i \right)^c 
						 \geq \left(\sum_{i} p_i q_i \right)^{2c-1}. 
	\end{equation*}
	For $c \leq 1$ we have that $x \mapsto x^c$ and $x \mapsto x^{2 - 1/c}$ are concave and the inequality is reversed.
\end{proof}
We are now ready to state our main lemma that lower bounds $\hat{f}(\alpha)$ in terms of $\hat{f}(0)$. This immediately implies Theorem~\ref{thm:revexphash}.
\begin{lemma} \label{lem:revexphash}
	For every $0 \leq \alpha < 1$ and for every distribution $\DSH$ over pairs of functions $h, g \colon \cube{d} \to R$, we have that $\hat{f}(\alpha) \geq \hat{f}(0)^{\frac{1+\alpha}{1-\alpha}}$.
\end{lemma}
\begin{proof}
For a function $h \colon \cube{d} \to R$ define its inverse image $h^{-1} \colon R \to 2^{\cube{d}}$ by $h^{-1}(i) = \{ \x \in \cube{d} \mid h(\x) = i \}$.
For a pair of functions $(h, g) \in \DSH$ and $i \in R$ we define $a_{h,i}, b_{g,i} \geq 0$ such that
$|h^{-1}(i)|/2^{d} = \exp(-a_{h,i}^{2}/2)$ and $|g^{-1}(i)|/2^{d} = \exp(-b_{g,i}^{2}/2)$.
For fixed $(h, g)$ define $\hat{f}_{h,g}(\alpha) = \Pr_{\text{$(\x, \y)$ $\alpha$-corr.}}[h(\x) = g(\y)]$. We obtain a lower bound on $\hat{f}(\alpha)$ as follows:
\begin{align*}
	\hat{f}(\alpha) &= \mathop{\E}_{(h, g) \sim \DSH} \left[ \sum_{i \in R} \Pr_{\text{$(\x, \y)$ $\alpha$-corr.}}[h(\x) = g(\y) = i] \right] \\
                 &\stackrel{(1)}{\geq} \mathop{\E}_{(h, g) \sim \DSH} \left[ \sum_{i \in R} \exp\left(-\frac{1}{2}\frac{a_{h,i}^{2} + 2\alpha a_{h,i}b_{g,i} + b_{g,i}^{2}}{1 - \alpha^2}\right) \right]  \\
                 &\stackrel{(2)}{\geq} \mathop{\E}_{(h, g) \sim \DSH} \left[ \sum_{i \in R} \exp\left(-\frac{1}{2}\frac{a_{h,i}^{2} + b_{g,i}^{2}}{1 - \alpha}\right) \right] \\
                 &\stackrel{(3)}{\geq} \mathop{\E}_{(h, g) \sim \DSH} \hat{f}_{h,g}(0)^\frac{1 + \alpha}{1 - \alpha} 
                 \stackrel{(4)}{\geq} \left( \mathop{\E}_{(h, g) \sim \DSH} \hat{f}_{h,g}(0) \right)^\frac{1 + \alpha}{1 - \alpha} \\ 
				 &= \hat{f}(0)^\frac{1 + \alpha}{1 - \alpha}. 
\end{align*}
Here, (1) is due to Theorem \ref{thm:expansionreverse}, (2)
follows from the simple fact that $a^2 + \alpha(a^2 + b^2) + b^2 \geq a^2 + 2\alpha a b + b^2$, (3)
follows from Lemma~\ref{lem:inequality} with $c = 1/(1-\alpha)$, 
and (4) follows from a standard application of Jensen's Inequality.
\end{proof}
\subsection{Extending the lower bound} \label{sec:lower_concentration}
We will now use Lemma~\ref{lem:revexphash} to together with concentration inequalities to obtain a lower bound on the $\rho_-$-value of DSH schemes with a CPF that is decreasing in $\simil_{H}(\x, \y)$. 
We introduce the following property:
\begin{definition} \label{def:sim_ins}
    Let $\DSH$ be a DSH family for $(X, \dist)$ with CPF $f$.
	We say that $\DSH$ is \emph{$(\alpha_{-}, \alpha_{+}, f_{-}, f_{+})$-decreasingly sensitive} (resp., \emph{$(\alpha_{-}, \alpha_{+}, f_{-}, f_{+})$-increasingly sensitive}) if it satisfies:  
	\begin{itemize}
		\item For $\alpha \leq \alpha_{-}$, we have $f(\alpha) \geq f_{-}$ (respectively, $f(\alpha) \leq f_{-}$);
		\item For $\alpha \geq \alpha_{+}$, we have $f(\alpha) \leq f_{+}$ (respectively, $f(\alpha) \geq f_{+}$).
	\end{itemize}
\end{definition}
We observe that if $\dist$ is a distance (resp., similarity) measure, then a decreasingly (resp., increasingly) sensitive DSH scheme corresponds to a standard LSH.
\begin{theorem} \label{thm:lower}
	Let $0 < \alpha_{-} < \alpha_{+} < 1$ be constants. 
	Then every $(\alpha_{-}, \alpha_{+}, f_{-}, f_{+})$-decreasingly sensitive family $\DSH$ for $(\cube{d}, \simil_{H})$ must satisfy 
	\begin{equation*}
		\rho_- = \frac{\log(1/f_{-})}{\log(1/f_{+})} \geq \frac{1 - \alpha_{+}}{1 + \alpha_{+} - 2\alpha_{-}} - O(\sqrt{\log(1/f_{+})/d}).
	\end{equation*}
\end{theorem}
In the statement of Theorem~\ref{thm:lower} we may replace the properties from Definition~\ref{def:sim_ins} that hold for every $\alpha \leq \alpha_{-}$ and every $\alpha \geq \alpha_{+}$ 
with less restrictive versions that hold in an $\varepsilon$-interval around $\alpha_{-}, \alpha_{+}$ for some $\varepsilon = o_{d}(1)$.
Furthermore, if we rewrite the bound in terms of relative Hamming distances $\delta$ and $\delta/c$ where $\delta, c$ are constants, we obtain a lower bound of $1/(2c - 1) - o_{d}(1)$ --- an expression that is familiar from known LSH lower bounds~\cite{motwani2007, andoni2016tight}.

We now prove Theorem \ref{thm:lower} proving an analogous result for a $(r, cr, p, q)$-increasingly sensitive family under Hamming distance.
Indeed, a $(\alpha_{-}, \alpha_{+}, f_{-}, f_{+})$-decreasingly sensitive family for $(\cube{d}, \simil_{H})$ is 
a $(r, cr, p, q)$-increasingly sensitive family for the space $(\cube{d}, \dist_{H})$, with $\alpha_{-}=1-r/d$ and $\alpha_{+}=1-cr/d$.
\begin{theorem} \label{thm:lowerhamming}
	For every constant $\varepsilon > 0$, we have that every $(r, cr, p, q)$-increasingly sensitive family $\ALSH$ for $\cube{d}$ under Hamming distance with $r \leq (1-\varepsilon)d/2$ must satisfy
	\begin{equation*}
		\rho(\ALSH) = \frac{\log 1/p}{\log 1/q} \geq \frac{1}{2c-1} - O(\sqrt{(c/r) \log(1/q)}).
	\end{equation*}
\end{theorem}
\begin{proof}
Given $\ALSH$ we define a distribution $\hat{\ALSH}$ over pairs of functions $\hat{h}, \hat{g} \colon \cube{\hat{d}} \to R$ where $\hat{d} \leq d$ remains to be determined.
We sample a pair of functions $(\hat{h}, \hat{g})$ from $\hat{\ALSH}$ by sampling $(h, g)$ from $\LSH$ 
and setting $\hat{h}(\x) = h(\x \circ \vect{1})$ and similarly $\hat{g}(\x) = g(\x \circ \vect{1})$ where $\vect{1}$ denotes the $(d - \hat{d})$-dimensional all-ones vector.
We will now turn to the process of relating $p$ to $\hat{p} = \hat{f}(0)$ and $q$ to $\hat{q} = \hat{f}(\alpha)$ for~$\hat{\ALSH}$.

Let $0 < \varepsilon_{p} < 1$ and set $\hat{d} = \lceil 2r/(1-\varepsilon_{p}) \rceil$. 
Then by applying standard Chernoff bounds we get 
\begin{equation*}
	\Pr_{(\x, \y) \, 0 \text{-corr.}}[\dist(\x, \y) \leq r] \leq \exp \left( -\frac{\varepsilon_{p}^{2}}{1-\varepsilon_{p}}\frac{r}{2} \right).
\end{equation*}
For convenience, define $\delta_p = \exp \left( -\frac{\varepsilon_{p}^{2}}{1-\varepsilon_{p}}\frac{r}{2} \right)$. Then 
$
	\hat{p} \geq (1 - \delta_p)p.
$

In order to tie $\hat{q}$ to $q$ we consider the probability of $\alpha$-correlated points having distance greater than $r/c$.
The expected Hamming distance of $\alpha$-correlated $(\x, \y)$ in $\hat{d}$ dimensions is $\hat{d}(1-\alpha)/2$.
We would like to set $\alpha$ such that the probability of the distance exceeding $r/c$ is small.
Let $X$ denote $\dist(\x, \y)$, then the standard Chernoff bound states that: 
\begin{equation*}
\Pr[X \geq (1+\varepsilon)\mu] \leq e^{-\varepsilon^2 \mu / 3}.
\end{equation*}
For a parameter $0 < \varepsilon_q < 1$ we set $\alpha$ such that the following is satisfied:
\begin{equation*}
	(1 + \varepsilon_q)\mu \geq (1 + \varepsilon_q) \frac{2r}{1 - \varepsilon_{p}}\frac{1 - \alpha}{2} \geq r/c.
\end{equation*}
This results in a value of $\alpha = 1 - \frac{1-\varepsilon_{p}}{1 + \varepsilon_{q}}\frac{1}{c}$ and we observe that
\begin{equation*}
	\delta_{q} \leq \exp(-\varepsilon_{q}^{2} \mu / 3) \leq \exp\left(- \frac{\varepsilon_{q}^{2}}{1+\varepsilon_{q}} \frac{r}{3c}\right).
\end{equation*}
It follows that
\begin{equation*}
	\hat{q} \leq (1-\delta_q)q + \delta_{q}.
\end{equation*}
Let us summarize what we know so far: 
\begin{align*}
\hat{p} &\geq (1-\delta_p)p \\
\hat{q} &\leq (1-\delta_q)q + \delta_q \leq q(1 + \delta_{q}/q) \\
0 &< \varepsilon_{p}, \varepsilon_{q} < 1 \\
\delta_{p} &\leq \exp \left( -\frac{\varepsilon_{p}^{2}}{1-\varepsilon_{p}}\frac{r}{2} \right), \quad \delta_{q} \leq \exp \left( -\frac{\varepsilon_{q}^{2}}{1+\varepsilon_{q}} \frac{r}{3c} \right) \\
\alpha &= 1 - \frac{1-\varepsilon_{p}}{1 + \varepsilon_{q}}\frac{1}{c}, \quad \hat{q} \geq \hat{p}^{\frac{1+\alpha}{1-\alpha}}.
\end{align*}
We assume that $0 < q < p < 1$ and furthermore, without loss of generality we can assume that $q \leq 1/e$ due to the powering technique (see Lemma~\ref{lemma:transform}(a)).
In our derivations we also assume that $\delta_p \leq 1/2$ and $\delta_q \leq 1/(2e)$ such that $\delta_{q}/q \leq 1/2$. This will later be implicit in the statement of the result in big-O notation.
From our assumptions and standard bounds on the natural logarithm we are able to derive the following:
\begin{align}
\frac{\ln(1/p)}{\ln(1/q)} &\geq \frac{\ln(1 - \delta_{p}) \ln(1/\hat{p})}{\ln(1/q)} \geq \frac{\ln(1/\hat{p})}{\ln(1/q)} - 2\delta_{p} \notag\\ 
		&\geq \frac{\ln(1/\hat{p})}{\ln(1 + \delta_{q}/q) + \ln(1/\hat{q})} - 2\delta_{p} \notag\\
		&\geq \frac{\ln(1/\hat{p})}{\ln(1/\hat{q})} \left(1 - \frac{\ln(1 + \delta_{q}/q)}{\ln(1/\hat{q})} \right) - 2\delta_{p} \notag\\
		&\geq \frac{\ln(1/\hat{p})}{\ln(1/\hat{q})} - \frac{\ln(1 + \delta_{q}/q)}{\ln(1/(1 + \delta_{q}/q)q)} - 2\delta_{p} \label{eq:hatineq} \\
		&\geq \frac{\ln(1/\hat{p})}{\ln(1/\hat{q})} - 2\delta_{q}/q - 2\delta_{p}.\notag
\end{align}
In equation \eqref{eq:hatineq} we use the statement itself combined with our assumptions on $p$ and $q$ to deduce that
\begin{equation*}
1 > \frac{\ln(1/p)}{\ln(1/q)} \geq \frac{\ln(1/\hat{p})}{\ln(1/\hat{q})}.
\end{equation*}
We proceed by lower bounding $\hat{\rho}$. Temporarily define $1 - \varepsilon' = \frac{1-\varepsilon_{p}}{1 + \varepsilon_{q}}$ and observe that
\begin{align*}
\frac{\ln(1/\hat{p})}{\ln(1/\hat{q})} &\geq \frac{1-\alpha}{1 + \alpha} = \frac{(1-\varepsilon')/c}{2 - (1-\varepsilon')/c} \\
&\geq \frac{1}{2c-1} - \frac{\varepsilon'}{(2c-1)^2} - \frac{\varepsilon'}{2c-1}.
\end{align*}
We have that 
\begin{equation*}
\varepsilon' = 1 - \frac{1-\varepsilon_{p}}{1 + \varepsilon_{q}} = \frac{1 + \varepsilon_{q} - (1-\varepsilon_{p})}{1 + \varepsilon_{q}} \leq \varepsilon_{q} + \varepsilon_{p},
\end{equation*}
and combining these bounds results in  
\begin{equation*}
\frac{\ln(1/p)}{\ln(1/q)} \geq \frac{1}{2c-1}- 2(\varepsilon_{q} + \varepsilon_{p} - \delta_{q}/q - \delta_{p}).
\end{equation*}
We can now set $\varepsilon_{q} = \varepsilon_{p} = K \cdot \sqrt{(c/r) \ln (1/q)}$ for some universal constant $K$ to obtain Theorem~\ref{thm:lower}.
\end{proof}

\subsection{Lower bound for asymmetric LSH}
We can re-apply the techniques behind Lemma \ref{lem:revexphash} and Theorem \ref{thm:lower} to state similar results in the other direction 
where for $\alpha_{-} < \alpha_{+}$ we are interested in upper bounding $f(\alpha_{+})$ as a function of $f(\alpha_{-})$.
This is similar to the well-studied problem of constructing LSH lower bounds and our results match known LSH bounds~\cite{motwani2007,andoni2016tight}, 
indicating that the asymmetry afforded by $\DSH$ does not help us when we wish to construct similarity-sensitive families with monotonically increasing CPFs.
Implicitly, this result already follows from the space-time tradeoff lower bounds for similarity search shown independently by Andoni et al.~\cite{andoni2017optimal} and Christiani~\cite{christiani2017framework}. 
As with Lemma~\ref{lem:revexphash}, the following theorem by O'Donnell~\cite{odonnell2014analysis} is the foundation of our lower bounds.
\begin{theorem}[Gen.\ Small-Set Expansion]\label{thm:expansion}
Let $0 \leq \alpha \leq 1$. Let $A, B \subseteq \cube{d}$ have volumes $\exp(-a^{2}/2)$, $\exp(-b^{2}/2)$ and assume $0 \leq \alpha b \leq a \leq b$. Then,
\begin{equation*}
	\Pr_{\substack{(\x, \y) \\ \alpha \text{-corr.}}}[\x \in A, \y \in B] \geq \exp\left(-\frac{1}{2}\frac{a^2 - 2\alpha a b + b^2}{1 - \alpha^2}\right).
\end{equation*}
\end{theorem}
\begin{lemma} \label{lem:exphash}
	For every $0 \leq \alpha < 1$ and for every distribution $\DSH$ over pairs of functions $h, g \colon \cube{d} \to R$, we have $\hat{f}(\alpha) \leq \hat{f}(0)^{\frac{1-\alpha}{1+\alpha}}$.
\end{lemma}
We are now ready to state the corresponding result for similarity-sensitive families.
\begin{theorem} \label{thm:lower2}
	Let $0 < \alpha_{-} < \alpha_{+} < 1$ be constants. 
	Then every $(\alpha_{-}, \alpha_{+}, f_{-}, f_{+})$-sensitive family $\DSH$ for Hamming space $(\cube{d}, \simil_{H})$ must satisfy 
	\begin{equation*}
		\frac{\log(1/f_{+})}{\log(1/f_{-})} \geq \frac{1 - \alpha_{+}}{1 + \alpha_{+} - 2\alpha_{-}} - O(\sqrt{\log(1/f_{-})/d}).
	\end{equation*}
\end{theorem}

%% file: algo.tex


\subsection{Anti-LSH construction in Hamming space}\label{sec:hamming}
Bit-sampling~\cite{IndykM98} is one of the simplest LSH families for Hamming space, yet gives optimal $\rho_+$-values in terms of the approximation factor~\cite{o2014optimal}.
Its CPF is $f(t)=1-t$, where $t$ is the relative Hamming distance.
By using a function pair $(\x\mapsto x_i, \x\mapsto 1-x_i)$ where $i\in \{1,\dots,d\}$ is random, we get a simple asymmetric DSH family for Hamming space whose CPF $f(t)=t$ is monotonically increasing in the relative Hamming distance.
We refer to this specific family as \emph{anti bit-sampling}. 
For anti bit-sampling, we get that $\rho_-= \ln f(r) / \ln f(r/c) = \Omega(1/\ln c)$ as soon as the relative Hamming distance $r \in [0,1]$ is smaller than $1/e$.
Perhaps surprisingly, anti bit-sampling is not optimal and a better result, with $\rho_-=O(1/c)$, follows by the anti-LSH schemes based on cross-polytope hashing and filters for the unit sphere in the following subsection. Similarly, the DSH construction for Euclidean space in section~\ref{sec:antilshEucl} gives a value of $O(1/c)$ for $\rho_-$.

\subsection{A DSH construction in Euclidean space}\label{sec:antilshEucl}
A simple and elegant  DSH family in Euclidean space is given by  a natural  extension of the LSH family introduced by Datar et al. \cite{Datar04}, where we project a point onto a line and split this line up into buckets.
Let $k$ and $w$ be two suitable parameters to be chosen below.
Consider the family $\mathcal{R}_{k, w}$ of pairs of functions $(h, g)$ defined in equation \eqref{eq:E2LSH++},
indexed by a uniform real number $b \in [0, w]$ and a $d$-dimensional random Gaussian vector $\a \sim \mathcal{N}^d(0,1)$.
We have the following result whose proof is provided in Appendix~\ref{app:proof:e2antilsh}:

\begin{theorem}
    Let $r_-$ and $r$ be two real values such that $0<r_-<r$, and let $c=r/r_-$.  Then there exists
    a constant $w = w(c)$ such that for each $k$ the family $\mathcal{R}_{k, w}$ satisfies $$\rho_- = \frac{\ln (1/f(r))}{\ln (1/f(r_-))} = \frac{1}{c^2} \left( 1  + O_k(1/k)\right).$$
    \label{thm:euclidean:anti:lsh}
\end{theorem}
\vspace{-1.5em}
The proof uses that for $\a \sim \mathcal{N}^d(0,1)$ the inner product 
$\langle \a, (\x - \y)\rangle$ is distributed as $\mathcal{N}(0, \Delta)$ for two points $\x$ 
and $\y$ at distance $\Delta$. For the hash values of $\x$ and $\y$ to collide,
the inner product must roughly lie in the interval $[(k - 1) w, (k + 1)w]$. Because we are 
free to choose $k$ and $w$, we can move the interval into the tail of the $\mathcal{N}(0, \Delta)$ distribution, where the target distances $r$ and $r/c$ have quadratic influence in the exponents.


%% file: constructions.tex
So far we have focused our attention on constructions with monotone CPFs, which just represent one kind of  DSH schemes.
It is natural to wonder if more advanced CPFs can be obtained.
In this section, we provide some results in this direction by describing two constructions yielding a wide class of CPFs.
We remark that these general constructions do not seem to provide any improved constructions in the monotone case.

\paragraph{Angular similarity functions}
We say that $\text{sim}\colon [-1,1] \rightarrow [0,1]$ is an \emph{LSHable angular similarity function} if there exists an hash family $\mathcal{S}$ with collision probability function $\text{sim}(\ip{\x}{\y})$ for each $\x, \y \in \sphere{d}$.
For example, the function $\text{sim}(t) = 1 - \arccos(t) / \pi$ is LSHable using the \emph{SimHash} construction of Charikar~\cite{Charikar02}.

Valiant~\cite{valiant2015finding} described  a pair of mappings $\varphi^{\mathcal P}_1,\varphi^{\mathcal P}_2\colon {\mathbb{R}}^d \rightarrow {\mathbb{R}}^D$, where $D=O(d^k)$, such that $\varphi^{\mathcal P}_1(\x)\cdot \varphi^{\mathcal P}_2(\y) = \pp{\ip{\x}{\y}}$, for any polynomial $\pp{t}=\sum_{i=0}^k a_i t^i$.
By leveraging this construction we get the following result (with proof provided  in Appendix~\ref{app:angfunc}).
\begin{theorem}\label{thm:distribution-sensitive}
	Suppose that $\text{sim}$ is an LSHable angular similarity function and that the polynomial $\pp{t}=\sum_{i=0}^k a_i t^i$ satisfies $\sum_{i=0}^k |a_i|$ = 1.
    Then there exists a distribution over pairs $(h,g)$ of functions such that for all $\x,\y\in \sphere{d}$, $\Pr[h(\x)=g(\y)] = \textnormal{sim}(\pp{\ip{\x}{\y}})$.
\end{theorem}

The computational cost of a naïve implementation of the proposed scheme may be prohibitive when $d^k$ is large.
However, by using the so-called \emph{kernel approximation} methods~\cite{pham2013fast}, we can in near-linear time compute approximations $\hat{\varphi}^{\mathcal P}_1(\x)$ and $\hat{\varphi}^{\mathcal P}_2(\y)$ that satisfy $\hat{\varphi}^{\mathcal P}_1(\x)\cdot \hat{\varphi}^{\mathcal P}_2(\y) = \pp{\ip{\x}{\y}}\pm \varepsilon$ with high probability for a given approximation error $\varepsilon > 0$.

\paragraph{Hamming distance functions}
It is natural to wonder which CPFs can be expressed as a function of the relative Hamming distance $d_h(\x,\y)$.
A first answer follows by using the anti bit-sampling approach from section~\ref{sec:algos} together with Lemma~\ref{lemma:transform}. This gives a scheme for matching any polynomial $\pp{t}=\sum_{i=0}^k a_i t^i$ that satisfies $\sum_{i=0}^k a_i = 1$ and $a_i > 0$ for each $i$.

In this section, we provide another construction that  matches, up to a scaling factor $\Delta$,  any polynomial $\pp{t}$ having no roots with a real part in $(0,1)$.
The scaling factor depends only on the roots of the polynomial. 
We claim that such a factor is unavoidable in the general case: indeed, without $\Delta$, it would be possible to match the CPF $1-t^2$ for Hamming space, which implies $\rho\leq 1/c^2$ in contradiction with the lower bound $1/c$ in~\cite{o2014optimal}. 
(Nevertheless, it is an open question to assess how tight $\Delta$ is.)
We have the following result that is proven in Appendix~\ref{sec:hamfunc}:

\begin{theorem}\label{thm:hamcontstr}
Let $\pp{t}=\sum_{i=0}^k a_i t^i$, $Z$ be the multiset of roots of $\pp{t}$, and $\psi\leq k$ be the number of roots with negative real part. 
Then there exists a DSH family with collision probability 
$\Pr(h(\x)=g(\y))=\pp{d_h(\x,\y)}/\Delta$ with  $\Delta = a_k 2^\psi \prod_{z\in Z, |z|>1} |z_i|$. 
\end{theorem}

The construction exploits the factorization  $\pp{t}=a_k \prod_{z\in Z}(t-z)$  and consists of a combination of $|Z|$ variations of bit-sampling and anti bit-sampling. 
Although the proposed scheme may not reach the $\rho$-value given by the polynomial $\pp{t}$, it can be used for estimating $\pp{d_H(\x,\y)}$ since the scaling factor is constant and only depends on the polynomial.

Finally, we observe that our scheme can be used to approximate any function $f(t)$ that can be represented with a Taylor series: indeed, it is sufficient to truncate the series to the  term that gives the desired approximation, and then to apply our construction to the resulting truncated polynomial.

%% file: applications.tex
\subsection{Hyperplane queries and annulus search}

\emph{Approximate annulus search} is the problem of finding a point in the set $P$ of data points with distance in an interval $[r_-,r_+]$ from a query point. On the unit sphere, hyperplane queries are a special case of annulus queries where we want to find a point with inner product close to $0$ to a query point. This type of search has applications in machine learning (see~\cite{Jain_NIPS10,Liu_ICML12}).

The ad hoc solution for this problem \cite{pagh2017approximate} in Euclidean space works by first building an LSH data structure  
that aims to retrieve points at distance at most $r$ while filtering out points at distance at least $r_+$. In each repetition and in each bucket of the hash table, one builds a data structure that is set up to filter away points at distance at least $r_-$ while preserving points at distance at least $r$. The latter filtering can be thought of as applying an anti-LSH in each bucket.  
For $r_- = r/c$ and $r_+ = cr$ for a $c \geq 1$,  the construction in \cite{pagh2017approximate} answers queries in time $\tilde{O}(dn^{\rho + 1/c^2})$ and the data structure uses space $\tilde{O}(n^{1 + \rho + 1/c^2} + dn)$, where $\rho$ describes the collision probability gap of the LSH that is used, disregarding logarithmic terms by using the $\tilde{O}(.)$ notation.

Having access to a DSH family with a CPF that peaks inside $[r_-,r_+]$ and is significantly smaller at the ends of the interval gives an LSH-like solution to this problem.

\begin{theorem}\label{thm:annulus-intro} 
Suppose we have a set $P$ of $n$ points, an interval $[r_-,r_+]$, a distance $r\in [r_-,r_+]$, 
and assume we are given a DSH scheme with a CPF $f$ that peaks inside $[r_-,r_+]$ and satisfies $f(r') \leq 1/n$ for all $r'\notin [r_-,r_+]$.
Then there exists a data structure that, given a query $\q$ for which there exists $\x\in P$ with $\dist(\q,\x)=r$, returns $\x'\in P$ with $\dist(\q,\x') \in [r_-,r_+]$ with probability at least $1/2$.
The data structure uses space $O(n^{1 + \rho^\ast}/f(r) + dn)$ and has query time $O(dn^{ \rho^\ast})$, where $\rho^\ast = \log (1/f(r))/\log n$.
\end{theorem}
\begin{proof}
The data structure is a straightforward adaptation of the construction of a near neighbor data structure using LSH.  
Associate  with each data point $\x$ and query point $\y$ the hash values $h(\x)$ and  $g(\y)$, where $(h, g)$ are independently sampled from the distance-sensitive family.
Store all points $\x \in S$ according to $h(\x)$ in a hash table. Let $\y$ be the query point and let $\x$ be a 
point at distance $r$. Compute $g(\y)$ and 
retrieve all the points from $S$ that have the same hash value. If a point within distance $[r_-,r_+]$ is among the  points, output one such point. 
We expect $\max\{f(r_-) n,f(r_+) n\} \leq 1$ collisions with points at distance 
at most $r_-$ or at least $r_+$.  The probability of finding $\x$ is at least $f(r)$. Thus, $L=1/f(r) \leq n^{\rho^\ast}$ 
repetitions suffice to retrieve $\x$ with constant probability $1/e$. If the algorithm retrieves more than $8L$ points, none of which is in the interval $[r_-,r_+]$, the algorithm terminates. By Markov's inequality, the probability that the algorithm retrieves $8L$ points, none of which is in the interval $[r_-,r_+]$, is at most $1/8$.  
\end{proof}
We note that the assumption $f(r_+), f(r_-) \leq 1/n$ in the theorem is not critical: the standard technique of \emph{powering} (see Lemma~\ref{lemma:transform}(a)) 
allows us to work with the CPF $f(x)^k$ for integer $k$, where $k$ is the smallest integer such that $f(x)^k \leq 1/n$.

We observe that  this data structure improves the trivial scanning solution when $\rho^\ast = \log (1/f(r))/\log n<1$, that is when  $f(r)>f(r_-)$ and $f(r)>f(r_+)$.
This is satisfied by \emph{unimodal} distance-sensitive hash families, that is when the CPF has a single maximum at $t^\ast$ and is decreasing for both $t \leq t^\ast$ and  $t \geq t^\ast$: as soon as $t^\ast$ lies in the interval $(r_-,r_+)$ we obtain a 
data structure with sublinear query time.

Obtaining a CPF that peaks inside of $[r_-, r_+]$ can be achieved by combining a standard LSH family $\mathcal{H}$ with a DSH family $\mathcal{A}$ that has an increasing CPF by means of powering each part. For example, when combining a bit-sampling and an anti bit-sampling family, concatenating $k_1$ bit-sampling and $k_2$ anti bit-sampling results in the CPF $f(t) = (1-t)^{k_1}t^{k_2}$. Setting $k_1 = k_2(1-t)/t$ results in $f$ peaking at distance $r$. In general, for the value $\rho^\ast$ in the
statement of the theorem, this approach yields a bound of $\rho^\ast \leq \rho_+ + \rho_-$, where $\rho_+$ and $\rho_-$ are the $\rho$-values of $\mathcal{H}$ and $\mathcal{A}$, respectively. This is essentially the same as the running time guarantee of the ad hoc data structure in \cite{pagh2017approximate} in Euclidean space.

In section~\ref{app:annulus}, we show how to combine the two monotonic constructions from section~\ref{sec:unit_sphere} on the unit sphere. A key ingredient of the construction is that we do not need the powering approach described above but can combine a single LSH with a single anti-LSH function and set the thresholds of each part accordingly. With respect to hyperplane queries under inner product similarity, the resulting construction 
allows us to search a point set $P$ of unit vectors for a vector approximately orthogonal to a query vector $\q$ in time $dn^{\rho^\ast + o(1)}$ for $\rho^\ast = \tfrac{1- \alpha^2}{1+\alpha^2}$, 
where we guarantee to return a vector $\x$ with $\langle \x , \q \rangle \in [-\alpha,\alpha]$ if an orthogonal vector exists.
When applied to search for approximately orthogonal vectors, our technique improves $\rho$-values of previous techniques, but the improvement is not surprising in view of recent progress in angular LSH, see e.g.~\cite{andoni2015practical}.
However, Theorem~\ref{thm:annulus_solution} in Section~\ref{app:annulus} supports searching a wide range of different annuli and not just the ones centered around vectors with zero correlation. As an additional result, we obtain a definition of an annulus in a space with bounded distances. 

\input{app_annulus}

\subsection{Spherical range reporting}
\emph{Approximate spherical range reporting}~\cite{AhleAP17} aims to report all points in $P$ within distance~$r$ from a query point. A common problem with LSH-based solutions for reporting all close points is that 
the CPF is monotonically decreasing starting with collision probability very 
close to $1$ for points that are very close to the query point. On the other hand, 
many repetitions are necessary to find points at the target distance $r$. This means 
that the algorithm retrieves many duplicates for solving range reporting problems.
The state-of-the-art data structure for range reporting queries~\cite{AhleAP17} requires $O\left( (1+|S^\ast|)(n/|S^\ast|)^\rho\right)$, where $S^\ast$ is the set of points at distance at most $r_+$.

CPFs that have a (roughly) fixed value in $[0,r]$ and then decrease rapidly to zero (so-called ``step-function CPFs'') yield data structures with good \emph{output sensitivity}. 
\begin{theorem}\label{thm:spherical-range}
Suppose we have a set $P$ of $n$ points and two distances $r < r_+$. 
Assume we are given a DSH scheme with CPF $f$ where $f(r')\leq 1/n$ for all $r' \geq r_+$, 
and let $f_\textnormal{min} = \inf_{t\in [0,r]} f(t)$, $f_\textnormal{max} = \sup_{t\in [0,r]} f(t)$.
Then there exists a data structure that, given a query $\q$, 
returns $S\subseteq \{ \x \in P  \mid  \dist(\q,\x) \leq r_+ \}$ such that for each $\x \in P$ with $\dist(\q,\x)\leq r$, $\Pr[\x \in S]>1/2$.
The data structure uses space $O(n^{1 + \rho^\ast} + dn)$ and the query has expected running time $O(dn^{\rho^\ast} + d|S|f_\textnormal{max}/f_\textnormal{min})$, 
where $\rho^\ast = \log (1/f_\textnormal{min})/\log (1/f(r_+))$.
\end{theorem}

\begin{proof}
    We assume that we build a standard LSH data structure as in the proof of Theorem~\ref{thm:annulus-intro} above.
We use $1/f_{\min}$ repetitions such that each point within distance $r$ is found with constant probability.
Each repetition will contribute $O\left(1+|S^\ast| f_{\max}\right)$ points in expectation.
Thus, the total cost will be $O\left((1+|S^\ast|f_{\max})/f_{\min}\right)$ from which the statement follows. 
\end{proof}
In short, Theorem~\ref{thm:spherical-range} provides a better analysis of the performance of a standard LSH 
data structure that takes into account the gap between $f_{\min}$ and $f_{\max}$. Again, the assumption $f(r_+) \leq 1/n$ in the theorem is not critical. 

In particular, the theorem shows that if we have a constant bound on $f_\text{max}/f_\text{min}$ the output sensitivity is optimal  
in the sense that the time to report an additional close point is $O(d)$ which is the time it takes to verify its distance to the query point. 
Such step-function CPFs are implicit in the linear space extremes of the space-time tradeoff techniques for near neighbor search~\cite{andoni2017optimal, christiani2017framework}.
To see why this is the case consider a randomized (list-of-points~\cite{andoni2017optimal}) data structure that solves the approximate near neighbor problem using linear space. 
The data structure stores each point in exactly one bucket.
During a query $L$ different buckets are searched.
The data structure has the property that near neighbors collide in at least one bucket with constant positive probability.
We can now construct a family $\DSH$ where we sample $(h, g) \sim \DSH$ by sampling a random (data-independent) data structure.
We set $h(\x)$ equal to the index of the single bucket that $\x$ would be stored in during an update, and $g(\x)$ is set to the index of a random bucket out of the $L$ buckets that would be searched during a query for $\x$. 
If the data structure guarantees finding an $r$-near neighbor, then the probability of collision is $\Theta(1/L)$ for points $\x, \y$ with $\dist(\x, \y) \leq r$. 
Even though this theoretical construction gives optimal output sensitivity, it is possible that a better value of $\rho^*$ can be obtained by allowing a higher space usage.

%

\subsection{Privacy-Preserving Distance Estimation}

Consider  a database consisting of private information, for example medical histories of patients, encoded as points in a space.
We would like to be able to search the database for points that are similar to a given query point $\q$ without revealing sensitive information about the data points, except whether there exists a point within a given distance $r$ from $\q$.
This is nontrivial even in the case of a single point, so we focus on the distance estimation problem: Is the distance between points $\q$ and $\x$ at most $r$ or not? Given $\x$ we would like to answer this while revealing as little information as possible about~$\x$.
Secure  multi-party computations can be used for this (see e.g.~\cite{goldreich2009foundations}), but such protocols are not practical in general.
In contrast, the intersection of two sets can be computed efficiently and in a private way that does not reveal anything about the items not in the intersection~\cite{DBLP:conf/eurocrypt/FreedmanNP04,de2010linear,pinkas2017}.

We are going to allow false positives and approximate false negatives as follows.
For an approximation factor $c>1$, and parameters $\varepsilon,\delta>0$:
\begin{itemize}
	\item If $\q$ and $\x$ have distance at most $r$, we say ``Yes'' with probability at least $1-\varepsilon$.
	\item If $\q$ and $\x$ have distance at least $cr$, we say ``No'' with probability at least $1-\delta$.
\end{itemize}

Our approach is to reduce this problem to private set intersection, as follows:
For a parameter $t$ to be chosen later, pick a DSH family $\mathcal{H}$ with a step-function CPF with  collision probability $\Theta(1/t)$ at distances in $[0,r]$, and let $\rho>0$ be the smallest constant such that the collision probability for distances larger than $cr$ is $O(t^{-1/\rho})$.
Without loss of generality we may assume that hash values are $O(\log t)$ bits (if not, hash them to this number of bits using universal hashing, increasing the collision probability only slightly).
Now generate a sequence of $O(t \log(1/\varepsilon))$ hash functions pairs $(h_1,g_1),(h_2,g_2),\dots \sim \mathcal{H}$, independently and consider the vectors $(h_1(\x),h_2(\x),\dots)$ and $(g_1(\x),g_2(\x),\dots)$.
By definition the expected (component-wise) intersection size of the two vectors, i.e., the expected number of hash collisions, is $O(\log(1/\varepsilon))$.
This means that the intersection of the two vectors reveals $O(\log(1/\varepsilon)\log t)$ bits of information about the vectors in expectation.
Also, by adjusting constants, the probability that the intersection size is $0$ when $\q$ and $\x$ are close is at most $\varepsilon$.
Thus, saying ``Yes'' if and only if the intersection is nonempty satisfies the first condition.
On the other hand, by a union bound the probability of a ``Yes''  when $\q$ and $\x$ have distance at least $cr$ is $\delta = O(t \log(1/\varepsilon) / t^{1/\rho})$.
Conversely, we need $t \approx (1/\delta)^{\rho/(1-\rho)}$ to have false positive probability $\delta$.

We observe that our approach has a stronger privacy constraint than the result in~\cite{Riazi16}. 
Indeed, the step-function DSH preserves privacy even if the $\q$ and $\x$ points are very close (e.g., $\q=\x$), since the collision probability is almost equal in the range $[0,r]$. On the other hand, a standard LSH, as the one adopted in~\cite{Riazi16}, has an high collision rate when the points are very close, revealing information on near points.

%

%% file: app_annulus.tex
\subsection{Annulus search on the unit sphere} \label{app:annulus}
We will construct a distance sensitive family $\DSH$ for solving the approximate annulus search problem.
Let $\DSH_{+}$ be parameterized by $t_{+}$ and let $\DSH_{-}$ be parameterized by $t_{-}$. 
To sample a pair of functions $(h, g)$ from $\DSH$ we independently sample a pair $(h_{+}, g_{+})$ from $\DSH_{+}$ and $(h_{-}, g_{-})$ from $\DSH_{-}$ 
and define $(h,g)$ by $h(\x) = (h_{+}(\x), h_{-}(\x))$ and $g(\x) = (g_{+}(\x), g_{-}(\x))$.

Let $f(\alpha)$ denote the CPF of $\DSH$. 
We would like to be able to parameterize $\DSH$ such that $f(\alpha)$ is somewhat symmetric around a unique maximum value of $\alpha$.
It can be verified from the definition of $\DSH_{+}$ that $p_{+}(-1) = 0$ which implies that $f(-1) = f(1) = 0$.
If we ignore lower order terms and define $\gamma > 0$ by $t_{-} = \gamma t_{+}$, then we can see that
\begin{equation*}
	\ln(1/f(\alpha)) \approx   \frac{1-\alpha}{1+\alpha}\frac{t_{+}^{2}}{2} + \frac{1+\alpha}{1-\alpha}\frac{\gamma^2 t_{+}^{2}}{2}.
\end{equation*}
For simplicity, temporarily define $a(\alpha) = (1-\alpha)/(1+\alpha) > 0$.
Given a fixed $\gamma$, the equation $a + \gamma^2 / a$ is minimized (corresponding to approximately maximizing $f(\alpha)$) when setting $a = \gamma$.
Let $\alpha_{\max} \in (-1, 1)$ and set $\gamma = a_{\max} = (1 -\alpha_{\max})/(1+\alpha_{\max})$. 
To find values $\alpha_{-} < \alpha_{\max} < \alpha_{+}$ where $\ln(1/f(\alpha_{-})) \approx  \ln(1/f(\alpha_{+}))$ 
note that this condition holds for every $s > 1$ when we set $a_{-} = sa_{\max}$ and $a_{+} = (1/s)a_{\max}$.
We therefore parameterize $\DSH$ by $\alpha_{\max} \in (-1, 1)$ and $t > 0$ and set $t_{+} = t$ and $t_{-} = (1-\alpha_{\max})/(1+\alpha_{\max})t_{+}$. 
By combining our bounds from Lemma~\ref{lem:cpfbounds} with the above observations we are able to obtain the following theorem 
which immediately yields a solution to the approximate annulus search problem. 
\begin{theorem} \label{thm:annulus}
	For every choice of $t > 0$ and constant $\alpha_{\max} \in (-1, 1)$ the family $\DSH$ satisfies the following: 
	For every choice of constant $s > 1$ consider the interval $[\alpha_{-}, \alpha_{+}]$ defined to contain every $\alpha$ such that
	$\frac{1}{s}\frac{1 - \alpha_{\max}}{1 + \alpha_{\max}} \leq \frac{1 - \alpha}{1 + \alpha} \leq  s\frac{1 - \alpha_{\max}}{1 + \alpha_{\max}}$, then  
	\begin{itemize}
		\item For $\alpha \in [\alpha_{-}, \alpha_{+}]$ we have that 
		\begin{equation*}
			f(\alpha) = \Omega\left((1/t^2) \exp\left(-(s + 1/s)\frac{1 - \alpha_{\max}}{1 + \alpha_{\max}}\frac{t^2}{2}\right)\right). 
		\end{equation*}
		\item For $\alpha \notin [\alpha_{-}, \alpha_{+}]$ we have that
		\begin{equation*}
		f(\alpha) = O\left((1/t^2) \exp\left(-(s + 1/s)\frac{1 - \alpha_{\max}}{1 + \alpha_{\max}}\frac{t^2}{2}\right)\right). 
		\end{equation*}
	\end{itemize}
	The complexity of sampling, storing, and evaluating a pair of functions $(h,g) \in \mathcal{D}$ is $O(d t^4 e^{t^2 / 2})$.
\end{theorem}
\noindent See Figure~\ref{fig:annuli} for a visual representation of the annulus for given parameters $\alpha_{\max}$ and $s$. 

\begin{figure}[t]
\centering
    \includegraphics[width=0.5\textwidth]{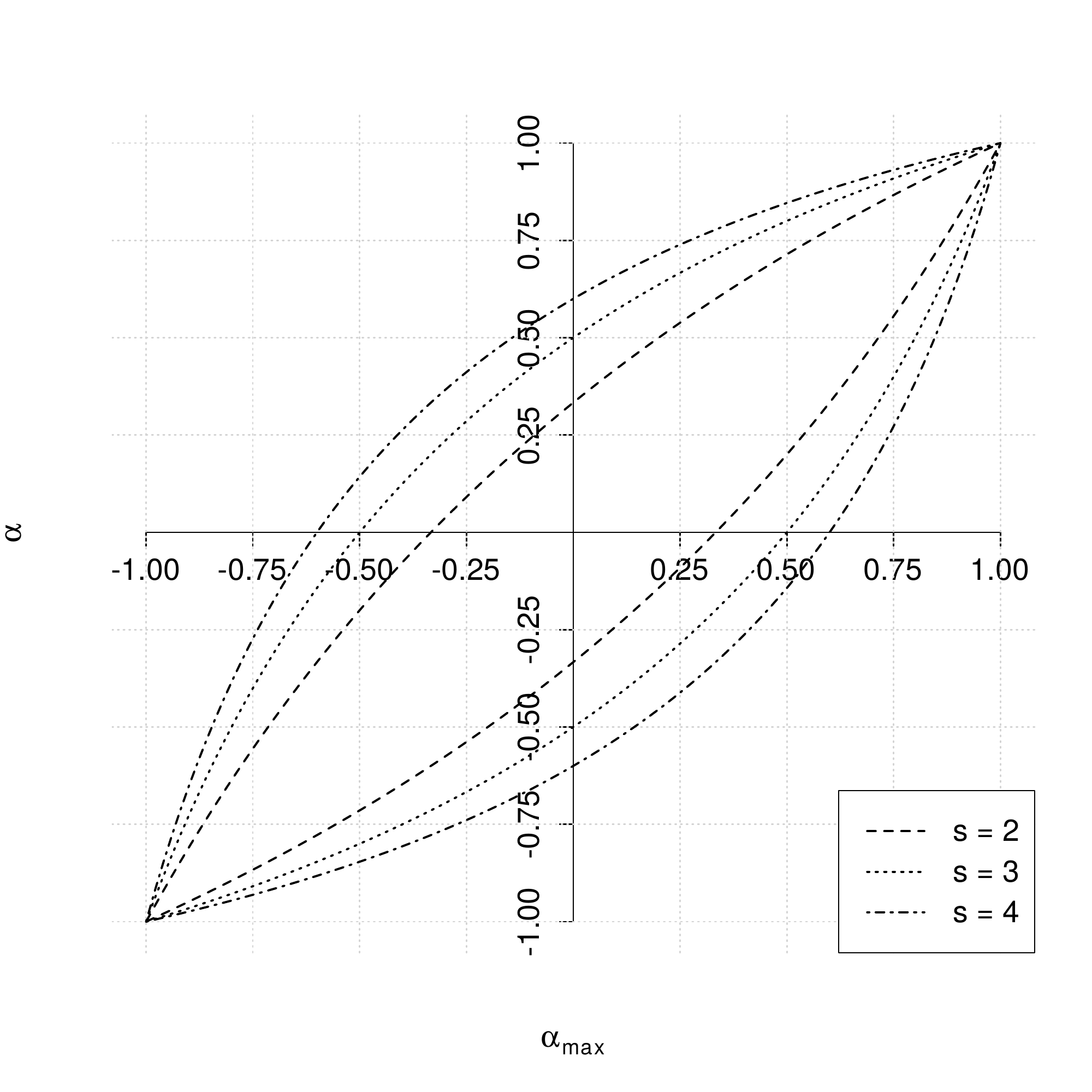}
\caption{Annuli as defined in Theorem \ref{thm:annulus} for every value of $\alpha_{\max}$ and $s = 2, 3, 4$.}
	\label{fig:annuli}
\end{figure}

We define an approximate annulus search problem for similarity spaces and proceed by applying 
Theorem \ref{thm:annulus} to provide a solution for the unit sphere, resulting in Theorem \ref{thm:annulus_solution}.
\begin{definition}
	Let $\beta_{-} < \alpha_{-} \leq \alpha_{+} < \beta_{+}$ be given real numbers.
	For a set $P$ of $n$ points in a similarity space $(X, \simil)$ a solution to the $((\alpha_{-}, \alpha_{+}), (\beta_{-}, \beta_{+}))$-annulus search problem 
	is a data structure that supports a query operation that takes as input a point $\x \in X$ and if there exists a point $\y \in P$ 
	such that $\alpha_{-} \leq \simil(\x, \y) \leq \alpha_{+}$ then it returns a point $\y' \in P$ such that $\beta_{-} \leq \simil(\x, \y') \leq \beta_{+}$.  
\end{definition}

\begin{theorem} \label{thm:annulus_solution}
	For every choice of constants $-1 < \beta_{-} < \alpha_{-} < \alpha_{+} < \beta_{+} < 1$ such that
	$\frac{1 - \alpha_{-}}{1 + \alpha_{-}}\frac{1 - \alpha_{+}}{1 + \alpha_{+}} = \frac{1 - \beta_{-}}{1 + \beta_{-}}\frac{1 - \beta_{+}}{1 + \beta_{+}}$
	we can solve the $((\alpha_{+}, \alpha_{-}), (\beta_{+}, \beta_{-}))$-annulus problem for $(\sphere{d}, \ip{\cdot}{\cdot})$ 
	with space usage $dn + n^{1 + \rho + o(1)}$ words and query time $dn^{\rho + o(1)}$ where
	\begin{equation*}
		\rho = \frac{c_{\alpha} + 1/c_{\alpha}}{c_{\beta} + 1/c_{\beta}} \leq \frac{2}{c + 1/c} 
	\end{equation*}
	and we define $1 < c_{\alpha} < c_{\beta}$ by $c_{\alpha} = \sqrt{\frac{1 - \alpha_{-}}{1 + \alpha_{-}} \big/ \frac{1 - \alpha_{+}}{1 + \alpha_{+}}}$, 
	$c_{\beta} = \sqrt{\frac{1 - \beta_{-}}{1 + \beta_{-}} \big/ \frac{1 - \beta_{+}}{1 + \beta_{+}}}$, and $c = c_{\beta}/c_{\alpha}$.
\end{theorem}

%% file: conclusion.tex

We have initiated the study of \emph{distance-sensitive hashing}, an asymmetric class of hashing methods that considerably extend the capabilities of standard LSH. We proposed   different constructions of such hash families and described some applications.  
Interestingly, DSHs provide an unique framework  for capturing problems that have been separately studied, like nearest neighbor~\cite{IndykM98}, finding orthogonal vectors~\cite{vijayanarasimhan2014hyperplane}, furthest point  query~\cite{pagh2017approximate}, and privacy preserving search~\cite{Riazi16}.

Though we settled some basic questions regarding what is possible using DSH, many questions remain.
Ultimately, one would like for a given space a complete characterization of the CPFs that can be achieved, with emphasis on extremal properties.
For example: For a CPF that has $f(x)=\Theta(\varepsilon)$ for $x\in [0,r]$, how small a value $\rho_{+}(c) = \log(f(r))/\log(f(cr))$ is possible outside of this range?
Additionally, our solution to the annulus problem works by combining an LSH and an anti-LSH family to obtain a 
unimodal family: While we know lower bounds for both, it is not clear whether combining 
them yields optimal solutions for this problem. 
Finally,  it is also of interest to consider other applications of DSH in privacy preserving search and in kernel density estimation (e.g.~\cite{Lambert1999}). 

\begin{acks}
	The authors thank Thomas D.~Ahle for insightful conversations. The research leading to these results has received funding from the \grantsponsor{SP4496}{European Research Council under the European Union's 7th Framework Programme (FP7/2007-2013)}{} / ERC grant agreement no.~\grantnum{SP4496}{614331}.
	Silvestri has also been supported by project \grantnum{pd}{SID2017} of the \grantsponsor{pd}{University of Padova}{}. BARC, Basic Algorithms Research Copenhagen, is supported by the \grantsponsor{barc}{VILLUM Foundation}{} grant \grantnum{barc}{16582}.
\end{acks}

%% file: app_monotone_constructions.tex

\section{Monotone DSH constructions}

%
%
%

\subsection{Optimal monotone DSH for the unit sphere}\label{app:tailbounds}
We initially bound the CPF of the family $\DSH_{+}$, which translates to a bound for $\DSH_{-}$ through the following observation:
\begin{lemma} \label{lem:symmetry}
	Given $\DSH_{+}$ and $\DSH_{-}$ with identical parameters, we have that $f_{+}(\alpha) = f_{-}(-\alpha)$.
\end{lemma}
\begin{proof}
	A bivariate normally distributed variable with correlation $\alpha$ can be represented as a pair $(X, Y)$ with $X = Z_1$ and 
	$Y = \alpha Z_1 + \sqrt{1-\alpha^2} Z_2$ where $Z_1, Z_2$ are i.i.d. standard normal.
	By the symmetry of the standard normal distribution around zero it is straightforward to verify that 
	$\Pr[Z_1 \geq t \land \alpha Z_1 + \sqrt{1-\alpha^2}Z_2 \geq t] =  \Pr[Z_1 \geq t \land -\alpha Z_1 + \sqrt{1-\alpha^2}Z_2 \leq -t]$. 
\end{proof}

The collision probability for $(h,g) \sim \DSH_{+}$ depends only on the inner product $\alpha = \ip{\x}{\y}$ between the pair of points being evaluated and is given by
\begin{equation*}
	\Pr[h(\x) \leq m \lor g(\y) \leq m] \frac{\Pr[\ip{\z}{\x} \geq t \land \ip{\z}{\y} \geq t]}{\Pr[\ip{\z}{\x} \geq t \lor \ip{\z}{\y} \geq t]}.
\end{equation*}
To see why this is the case first note that it is only possible that $h(\x) = g(\y)$ in the event that $h(\x) \leq m \lor g(\y) \leq m$.
Conditioned on this happening, consider the first $i$ such that either $\ip{\z_i}{\x} \geq t$ or $\ip{\z_i}{\y} \geq t$.
Now the probability of collision is given by $\Pr[\ip{\z}{\x} \geq t \land \ip{\z}{\y} \geq t \vert \ip{\z}{\x} \geq t \lor \ip{\z}{\y} \geq t]$. 

To bound the CPF of $\DSH_{+}$ and   $\DSH_{-}$, we use the following tail bounds for the standard normal distribution and the tail bounds by Savage \cite{savage1962} for the bivariate standard normal distribution.

\begin{lemma}[Follows Szarek \& Werner \cite{szarek1999}] \label{lem:univariatebounds}
	Let $Z$ be a standard normal random variable. Then, for every $t \geq 0$ we have that
	\begin{equation*}
		\frac{1}{\sqrt{2\pi}}\frac{1}{t+1}e^{-t^{2}/2} \leq \Pr[Z \geq t] \leq \frac{1}{\sqrt{2\pi}}\frac{1}{t}e^{-t^{2}/2}. 
	\end{equation*}
\end{lemma}

\begin{lemma}[Savage {\cite{savage1962}}] \label{lem:savage}
	Let $\alpha \in (-1,1)$ and let $Z_1, Z_2 \sim \mathcal{N}(0,1)$. Define $X_1 = Z_1$ and $X_2 = \alpha Z_1 + \sqrt{1-\alpha^2}Z_2$.
Then, for every $t > 0$ we have that


\begin{align*}
	&\left(1-\frac{(2-\alpha)(1+\alpha)}{1-\alpha}\frac{1}{t^2}\right)\frac{1}{2 \pi t^2} \frac{(1 + \alpha)^2} {\sqrt{1-\alpha^2}} \exp\left(-\frac{t^2}{1+\alpha}\right) \\
	&< \Pr[X_1 \geq t \land X_2 \geq t] < \frac{1}{2 \pi t^2} \frac{(1 + \alpha)^2} {\sqrt{1-\alpha^2}} \exp\left(-\frac{t^2}{1+\alpha}\right)
\end{align*}

\end{lemma}

\begin{corollary}
	By symmetry of the normal distribution the Lemma \ref{lem:savage} bounds apply to $\Pr[X_1 \geq t \land X_2 \leq -t]$ when we replace all occurrences of $\alpha$ with $-\alpha$.
\end{corollary}

We are now ready to bound the CPF for  $\DSH_{+}$.

\begin{lemma} \label{lem:cpfbounds}
	For every $t > 0$ and $\alpha \in (-1, 1)$ the family $\DSH_{+}$ satisfies 
	\begin{align*}
		f_{+}(\alpha) &< \bar{f}_{+}(\alpha) := \frac{1}{\sqrt{2\pi}}\frac{t + 1}{t^2}\frac{(1+\alpha)^2}{\sqrt{1-\alpha^2}}\exp\left(-\frac{1-\alpha}{1+\alpha}\frac{t^2}{2}\right), \\
		f_{+}(\alpha) &> \left(1-\frac{(2-\alpha)(1+\alpha)}{1-\alpha}\frac{1}{t^2}\right) \frac{t}{t + 1} \bar{f}_{+}(\alpha) - 2e^{-t^3}. 
	\end{align*}
	The complexity of sampling, storing, and evaluating a pair of functions $(h,g) \in \DSH_{+}$ is $O(d t^4 e^{t^2 / 2})$.
\end{lemma}
\begin{proof}
We proceed by deriving upper and lower bounds on the collision probability.
\begin{align*}
	f_{+}(\alpha) &\leq \frac{\Pr[\ip{\z}{\x} \geq t \land \ip{\z}{\y} \geq t]}{\Pr[\ip{\z}{\x} \geq t]} \\
					   &\leq \frac{1}{\sqrt{2\pi}}\frac{t + 1}{t^2}\frac{(1+\alpha)^2}{\sqrt{1-\alpha^2}}\exp\left(-\frac{1-\alpha}{1+\alpha}\frac{t^2}{2}\right). 
\end{align*}
We derive the lower bound in stages.
\begin{align*}
	&\Pr[h(\x) = g(\y)] \\ 
	&\geq (1 - \Pr[h(\x) > m])
					  \frac{\Pr[\ip{\z}{\x} \geq t \land \ip{\z}{\y} \geq t]}{\Pr[\ip{\z}{\x} \geq t \lor\ip{\z}{\y} \geq t]} \\
	&\geq \frac{\Pr[\ip{\z}{\x} \geq t \land \ip{\z}{\y} \geq t]}{2 \Pr[\ip{\z}{\x} \geq t]} - \Pr[h(\x) > m]. 
\end{align*}
The first part is lower bounded by
\begin{equation*}
\begin{split}
	&\frac{\Pr[\ip{\z}{\x} \geq t \land \ip{\z}{\y} \geq t]}{2 \Pr[\ip{\z}{\x} \geq t]} \geq 
	\left(1-\frac{(2-\alpha)(1+\alpha)}{1-\alpha}\frac{1}{t^2}\right) \cdot \\
	&\qquad \frac{1}{2 \sqrt{2 \pi}} \frac{1}{t} \frac{(1 + \alpha)^2} {\sqrt{1-\alpha^2}} \exp\left(-\frac{1-\alpha}{1+\alpha}\frac{t^2}{2}\right). 
\end{split}
\end{equation*}
The probability of not being captured by a projection depends on the number of projections $m$. 
In order to make this probability negligible we can set $m = \lceil 2t^3/p' \rceil$ where $p'$ denotes the lower bound from Lemma \ref{lem:univariatebounds}.
\begin{align*}
	\Pr[h(\x) > m] 				  &\leq (1 - \Pr[\ip{\z}{\x} \geq t])^m 
							      \leq (1 - p')^{2t^3/p'} \leq e^{-2t^3}. 
\end{align*}
The bound on the complexity of sampling, storing, and evaluating a pair of functions $(h,g) \in \DSH_{+}$ follows from having $m = \lceil 2t^3/p' \rceil = O(t^4 e^{t^{2}/2})$ standard normal projections of length $d$ to be sampled, stored, and evaluated.
\end{proof}

Combining the above ingredients we get the following results, which implies  Theorem~\ref{thm:anti} by Lemma \ref{lem:symmetry}.

\begin{theorem} 
	For every $t > 1$ there exists a distance-sensitive family $\DSH_{+}$ for $(\sphere{d}, \ip{\cdot}{\cdot})$ 
	with a CPF $f$ such that for every $\alpha \in (-1,1)$ satisfying $|\alpha| < 1 - 1/t$ we have that
	\begin{equation}
		\ln(1/f(\alpha)) = \tfrac{1 - \alpha}{1 + \alpha}\tfrac{t^2}{2} + \Theta(\log t) .
	\end{equation}
Furthermore, the CPF of $\DSH_{+}$ is monotonically increasing, and the complexity of sampling, storing, and evaluating $(h,g) \in \DSH_{-}$ is $O(d t^4 e^{t^2 / 2})$.
\end{theorem}

    A more careful analysis of the collision probabilities is required in order to combine the families $\DSH_{-}$ and $\DSH_{+}$ to form a unimodal family that can be used to solve the annulus search problem, see Theorem~\ref{thm:annulus-intro}.
These results are stated in Appendix \ref{app:annulus}.

%% file: app_lower_sphere.tex


%% file: app_euclidean_antilsh.tex
\section{DSH for Euclidean space}\label{app:proof:e2antilsh}

\begin{proof}[Proof of Theorem~\ref{thm:euclidean:anti:lsh}]
For the sake of simplicity we assume $r=1$ in the analysis (otherwise it is enough to scale down vectors accordingly).
Let $\x$ and $\y$ be two points in $\mathbb{R}^d$ with distance $\Delta$. We know that for $\a \sim \mathcal{N}^d(0,1)$ the inner product 
$\langle \a, (\x - \y)\rangle$ is distributed as $\mathcal{N}(0, \Delta)$. A necessary but not sufficient condition to have a collision between~$\x$ and~$\y$ is that  $\langle \a, (\x - \y)\rangle$ lies in the interval 
$[(k - 1) w, (k + 1)w]$. Now, if $t := \langle \a, (\x - \y)\rangle \in [(k -1)w, kw]$, then the random offset $b$ must lie in an interval of length $t - (k - 1) w$, 
putting $\langle \a, \x \rangle$ and $\langle \a, \y \rangle - (k - 1)w$ into different buckets. For the interval $[kw, (k + 1)w]$ similar observations show 
that $b$ has to be chosen in an interval of length $(k + 1)w - t$.
Let $\phi(t) = 1/\sqrt{2 \pi} e^{-t^2/2}$ be the density function of a standard normal random variable. Similarly to the calculations in \cite{Datar04}, the collision probability at distance $\Delta$ can be calculated as follows:
\begin{align*}
  \hspace{-1cm}  f&(\Delta) = \Pr\left(\left\lfloor\frac{\langle \a, \x \rangle + b}{w}\right\rfloor - \left\lfloor\frac{\langle \a, \y \rangle + b}{w}\right\rfloor = k\right)\\
              =& \int_{(k - 1) w}^{k w} \frac{\phi(t/\Delta)}{\Delta} \left( \frac{t}{w} - (k - 1)\right) dt  \\ 
               &\quad\quad+ \int_{k w}^{(k + 1) w} \frac{\phi(t/\Delta)}{\Delta} \left( k + 1 - \frac{t}{w}\right)  dt - \frac{\phi(kw/\Delta)}{\Delta} \\
    =& \frac{1}{\sqrt{2\pi}\Delta}\Biggl(\int_{(k - 1) w}^{k w} e^{-\frac{t^2}{2\Delta^2}} \left(\frac{t}{w} - (k - 1) \right) dt \\  &\quad\quad + \int_{k w}^{(k + 1) w} e^{-\frac{t^2}{2\Delta^2}} \left(k + 1 - \frac{t}{w}\right) dt - e^{-\frac{(kw)^2}{2\Delta^2}}\Biggr).
\end{align*}
We now proceed to upper bound $\rho^-$
by finding an upper bound on $f(1/c)$ and a lower bound on $f(1)$. 
Simple calculations give an upper bound of 
\begin{align}
    f(1/c) \leq \frac{2wc}{\sqrt{2 \pi}}e^{-(c(k - 1)w)^2/2}.
    \label{eq:upper}
\end{align}
For the lower bound, we only look at the interval $t \in [kw, (k + 1/2)w]$ and obtain the bound: 
\begin{align}
    f(1) &\geq \frac{1}{\sqrt{2\pi}} \int_{k w}^{(k + 1/2) w} e^{-\frac{t^2}{2}} \left(k + 1 - \frac{t}{w}\right) dt \notag\\
         &\geq \frac{w}{4\sqrt{2 \pi}} e^{-((k + 1/2)w)^2/2}.
    \label{eq:lower}
\end{align}
Now we multiply the ratio of the logarithms of the right-hand sides of \eqref{eq:upper} and \eqref{eq:lower} with $c^2$ and proceed to show that this term is bounded by $1 + O(1/k)$, which shows the result. In the following, we set 
$w$ such that $w \leq \sqrt{2 \pi}/(2c)$. We compute:
\begin{align*}
    &\frac{\ln\left(\frac{w}{4\sqrt{2 \pi}} e^{-((k + 1/2)w)^2/2}\right)}{\ln\left(\frac{2wc}{\sqrt{2 \pi}}e^{-(c(k - 1)w)^2/2}\right)} c^2 \\
&= \frac{-2\ln(\frac{w}{4\sqrt{2 \pi}}) + ((k + 1/2)w)^2}{-2 \ln(\frac{2wc}{\sqrt{2 \pi}})c^2 + (( k - 1)w)^2} \leq  \frac{-2\ln(\frac{w}{4\sqrt{2 \pi}}) + ((k + 1/2)w)^2}{(( k - 1)w)^2} \\
&=  \frac{(k + 1/2)^2}{( k - 1)^2} + O(1/k^2) = 1 + O(1/k). 
\end{align*}
\end{proof}

%% file: app_general_constructions.tex

\section{General constructions}

\subsection{Proof of Lemma~\ref{lemma:transform}}\label{app:transform}

\begin{proof}
We consider the transformation in~\cite{Chierichetti15} in the asymmetric setting.
Let $\x, \y$ be two arbitrary points from $X$. 
Part (a): Sample a pair $(h_i, g_i)$ from $\mathcal{D}_i$ for each $i \in \{1,\ldots,n\}$ and set $h(\x) = (h_1(\x), \ldots, h_n(\x))$ and $g(\y) = (g_1(\y), \ldots, g_n(\y))$. We observe that 
\begin{align*}
\Pr(h(\x) = g(\y)) &= \prod_{i = 1}^n \Pr(h_i(\x) = g_i(\y)) = \prod_{i = 1}^n f_i(\text{dist}(\x, \y)).
\end{align*}
Part (b): Pick an integer $i \in \{1, \ldots, n\}$ according to $\{p_i\}$ at random. Then sample a pair $(h_i, g_i)$ from $\mathcal{D}_i$. The hash function pair $(h, g)$ is given by $(i, h_i(\x))$ and $(i, g_i(\y))$.  We observe that 
\begin{align*}
    \Pr(h(\x) = g(\y)) &= \sum_{i = 1}^{n} p_i \Pr_{(h,g) \sim \mathcal{D}_i}(h(\x) = g(\y)) = 
\sum_{i = 1}^n p_i f_i(\text{dist}(\x, \y)). 
\end{align*}
\end{proof}

\subsection{Angular similarity function}\label{app:angfunc}
This section shows how to derive a distance sensitive scheme with collision probability $\text{sim}(\pp{\ip{\x}{\y}})$, when $\sum_{i=0}^k |a_i| = 1$.
Figure~\ref{fig:CPF} gives some examples of functions that can be obtained from Theorem~\ref{thm:distribution-sensitive} using SimHash~\cite{Charikar02}.

\begin{figure*}
\includegraphics[width=0.45\linewidth]{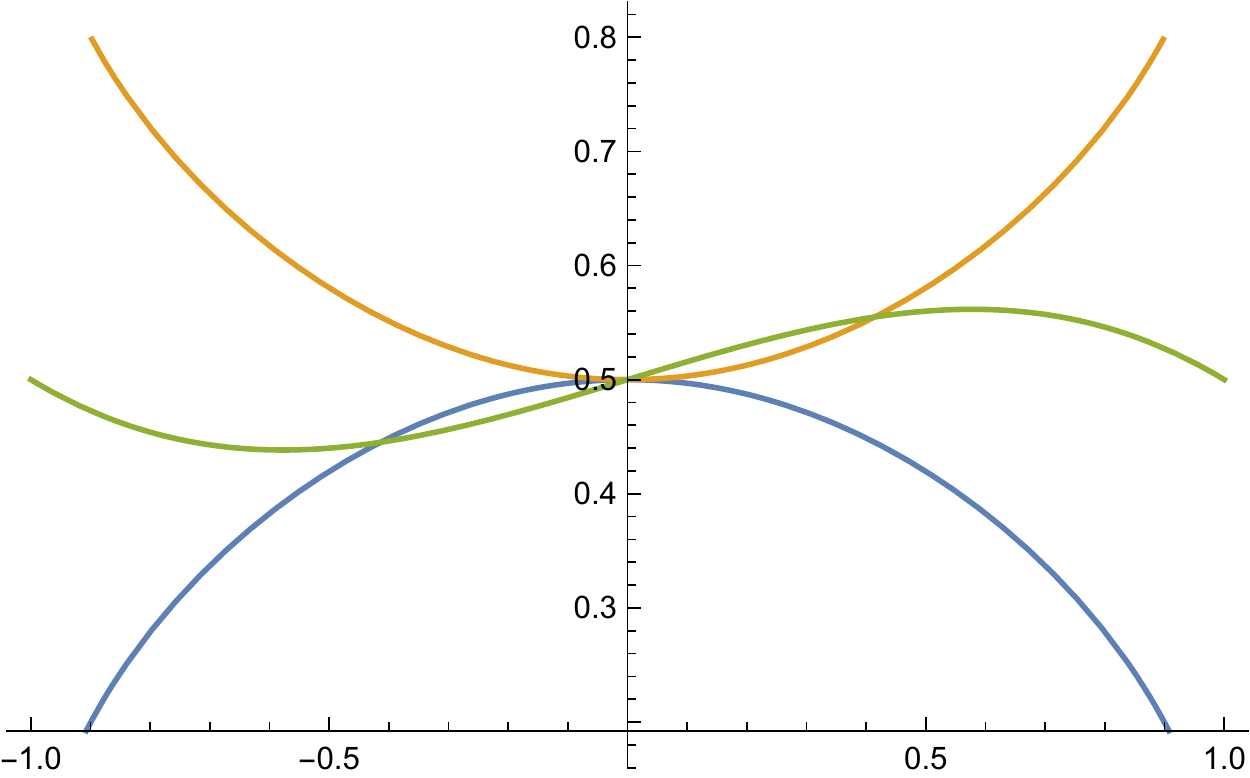}
\hfill
\includegraphics[width=0.45\linewidth]{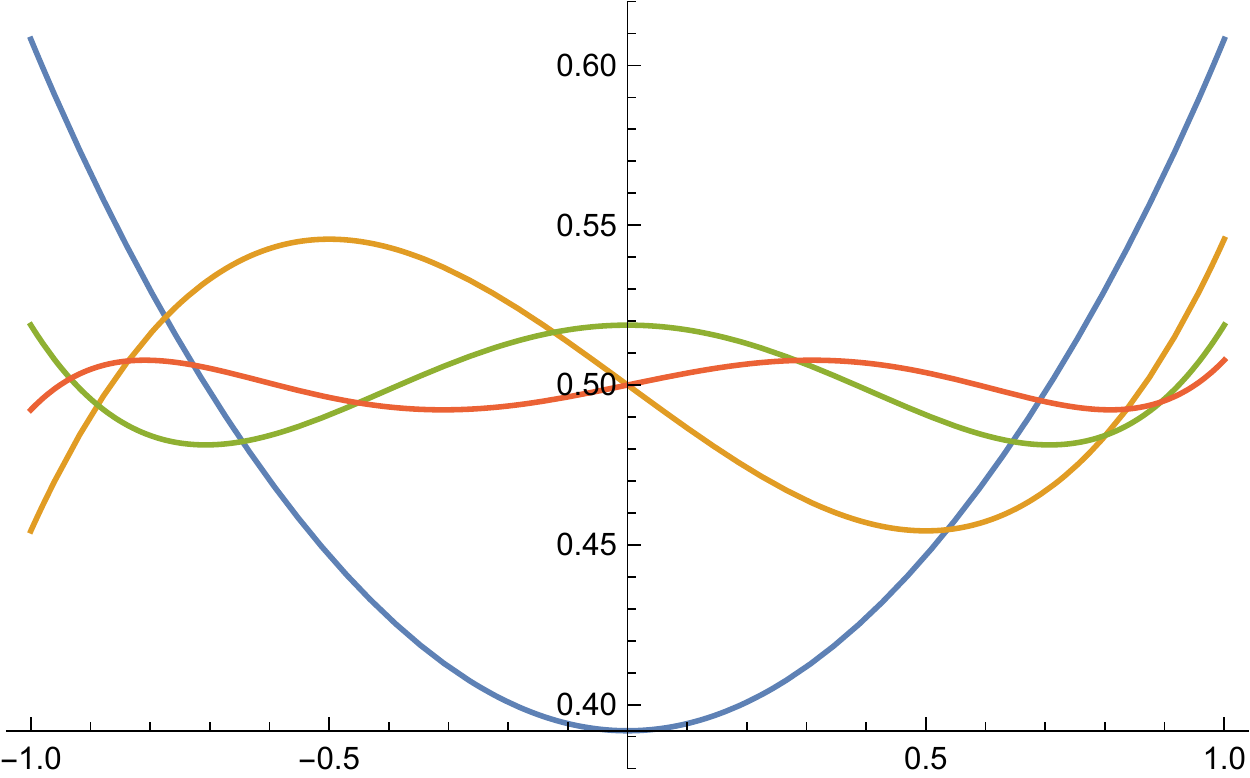}
\caption{Examples of collision probability functions obtained using Theorem~\ref{thm:distribution-sensitive}. The polynomials used are $t^2$, $-t^2$, $(-t^3 + t^2 - t)/3$ (left), and $(2 t^2 - 1)/3$, $(4 t^3 - 3 t)/7$, $(8 t^4 - 8 t^2 + 1)/17$, $(16 t^5 - 20 t^3 + 5 t)/41$ (right).}\label{fig:CPF}
\end{figure*}

\begin{proof}[Proof of Theorem~\ref{thm:distribution-sensitive}]
    Valiant~\cite{valiant2015finding} has shown how, for any real degree-$k$ polynomial $p$, to construct a pair of mappings $\varphi^p_1,\varphi^p_2: {\mathbb{R}}^d \rightarrow {\mathbb{R}}^D$, where $D=O(d^k)$, such that $\varphi^p_1(\x)\cdot \varphi^p_2(\y) = \pp{\ip{\x}{\y}}$.
	For completeness we outline the argument here:
	First consider the monomial $\pp{t} = a_k t^k$.
    For $\x\in {\mathbb{R}}^d$, let $\x^{(k)}$ denote the vector of dimension $d^k$ indexed by vectors ${\bf i} = (i_1,\dots,i_k) \in [d]^k$, where $\x^{(k)}_{\bf i} = \prod_{j=1}^k x_{i_j}$.
    It is easy to verify that $\ip{\x^{(k)}}{\y^{(k)}} = (\ip{\x}{\y})^k$ for all $\x,\y\in {\mathbb{R}}^d$.
    With this notation in place we can define $\varphi^p_1(\x) = \sqrt{|a_i|}\, \x^{(k)} $ and $\varphi^p_2(\y) = (a_i / \sqrt{|a_i|})\, \y^{(k)}$ which satisfy $\varphi^p_1(\x) \cdot \varphi^p_2(\x) = a_i (\ip{\x,\y})^k$.
	The asymmetry of the mapping is essential to allow a negative coefficient $a_k$.
	To handle an arbitrary real polynomial $\pp{t}=\sum_{i=0}^k a_i t^i$ we simply concatenate vectors corresponding to each monomial, obtaining a vector of dimension $\sum_{i=0}^k d^i = O(d^k)$.
	
    Observe that $||\x^{(k)}||_2^2 = \ip{\x^{(k)}}{ \x^{(k)}} = (\ip{\x,\x})^k = ||\x||_2^{2k}$.
	This means that for $||\x||_2^2 = 1$ we have $||\varphi^p_1(\x)||^2 = \sum_{i=0}^k \sqrt{|a_i|}^2 = 1$, using the assumption $\sum_{i=0}^k |a_i| = 1$.
	Similarly, we  have for  $||\y||_2^2 = 1$ that $||\varphi^p_2(\y)||^2 = \sum_{i=0}^k (a_i/\sqrt{|a_i|})^2 = \sum_{i=0}^k |a_i| = 1$.	
	Thus, $\varphi^p_1$ and $\varphi^p_2$ map $S^{d-1}$ to $S^{D-1}$.

	Our family $\mathcal{F}$ samples a function $s$ from the distribution $\mathcal{S}$ corresponding to $\text{sim}$ and constructs the function pair $(h,g)$ with $h(\x) = s(\varphi^p_1(\x))$, $g(\y) = s(\varphi^p_2(\y))$.
    Using the properties of the functions involved we have $$\Pr[h(\x)=g(\y)] = \text{sim}(\ip{\varphi^p_1(\x)}{\varphi^p_2(\y)}) = \text{sim}(\pp{\ip{\x}{\y}}).$$
\end{proof}

\subsection{Hamming distance functions}\label{sec:hamfunc}

\begin{proof}[Proof of Theorem~\ref{thm:hamcontstr}]
We initially assume that $a_0\neq 0$ (i.e., 0 is not a root of $\pp{t}$), and then remove this assumption at the end of the proof.
We recall that a root of $\pp{t}$ can appear with multiplicity larger than 1  and that, by the  complex conjugate root theorem, if $z=a + b i$ is a complex root then so is its conjugate $z'=a-b i$. 
We let $Z$ be the multiset containing the $k$ roots of $\pp{t}$, with $Z_{r+}$ and $Z_{r-}$ being the multiset of positive and negative real roots, respectively, and with $Z_c$ being the multiset consisting of pairs of conjugate complex roots. 
By factoring $\pp{t}$, we get:
\begin{align}\label{eq:pol_fact}
    \pp{t} &= a_k  \prod_{z\in Z} (t-z) = |a_k|  \prod_{z\in Z_{r+}} (z-t) \prod_{z\in Z_{r-}} (t+|z|) \cdot \notag\\  &\quad\quad\prod_{z=a+bi\in Z_c} (t^2-2at +a^2 + b^2),
\end{align}
where  the last step follows since $a_k  \prod_{z\in Z_{r+}} (z-t)= |a_k|  \prod_{z\in Z_{r+}} (t-z)>0$. Indeed, $\pp{t}$ is positive in $(0,1)$ and the multiplicative terms associated with complex and negative real roots are positive in this range; this implies that the remaining terms are positive as well.

We need to introduce scaled and biased  variations of bit-sampling or anti bit-sampling.
Anti-bit sampling with scaling factor $\alpha\in [0,1]$ and bias $\beta \in [0,1]$ has the CPF  $f(t)=\beta/2+\alpha t/2$ and is given by  randomly selecting one of following two schemes: 
(1) with probability $1/2$, the scheme is a standard hashing that maps data and query points to 0 with probability $\beta$, and otherwise to 0 and 1 respectively; 
(2) with probability $1/2$, the scheme is anti bit-sampling where the sampled bit is set to 0 with probability $1-\alpha$ on both data and query points, or kept unchanged otherwise.
Similarly, bit-sampling with scaling factor $\alpha\in [0,1]$ has the CPF $f(t)=(1-\alpha t)$ and is given by  using  bit-sampling, where the sampled bit is set to 0 with probability $1-\alpha$ on both data and query points. (We do not need a biased version of bit-sampling.)

We now assign  to each multiplicative term of~\eqref{eq:pol_fact} a scaled and biased  version of bit-sampling or anti bit-sampling as follows:
\begin{itemize}
\item \textbf{$z$ is real and $z<-1$.}
We assign to $z$  an anti bit-sampling with bias $1$ and scaling factor $1/|z|\leq 1$: the CPF is $S_1(t,z) = (1/2+t/(2|z|))$, and we have $(t+|z|) =  2|z| S_1(t,z)$.

\item \textbf{$z$ is real and $-1\leq z< 0$.} 
We assign to $z$  an anti bit-sampling with bias $|z|\leq 1$ and scaling factor $1$: the CPF is $S_2(t,z) = |z|/2+t/2$, and we have $(t+|z|) = 2 S_2(t,z)$.

\item \textbf{$z$ is real and $ z\geq 1$.} 
We assign to $z$  a bit-sampling with scaling factor $1/z\leq 1$: the CPF is $S_3(t,z) = (1-t/z)$, and we have $(t-z)=z S_3(t,z)$.

\item \textbf{$(z,z')$ are conjugate complex roots and $\textnormal{Real}(z)< -1$.}
Let $z=a+bi$ and $z'=a-bi$.
The assigned scheme has CPF $$S_4(t,z)=\left(\frac{b^2}{4(a^2+b^2)}+\frac{a^2}{a^2+b^2}\left(\frac{x}{2|a|}+\frac{1}{2}\right)^2\right)$$
and is obtained as follows:
with probability $b^2/(a^2+b^2)$, the scheme maps data and query points to 0 and 0 with probability 1/4, or to 0 and 1 with probability 3/4;
with probability $a^2/(a^2+b^2)$, the schemes consists of the concatenation of two anti bit-sampling with bias $1$ and scaling factor $1/|a|$.
Note that 
$t^2 - 2a t + a^2 + b^2=4(a^2+b^2) S_4(t,z)$.

\item  \textbf{$(z,z')$ are conjugate complex roots and $\textnormal{Real}(z)\geq 1$.}
The scheme is similar to the previous one where we use two bit-sampling with scaling factor $1/a$ instead of the anti bit-sampling. 
The CPF is $$S_5(t,z)=\left(\frac{b^2}{a^2+b^2}+\frac{a^2}{a^2+b^2}\left(1-\frac{x}{a}\right)^2\right),$$ and we get $t^2 - 2a t + a^2 + b^2=(a^2+b^2) S_5(t,z)$.

\item \textbf{$(z,z')$ are conjugate complex roots, $-1\leq \textnormal{Real}(z)\leq  0$, and $|z|=a^2+b^2 \geq 1$.}
We assign the following scheme with CPF $$S_6(t,z)=\left(\frac{x^2}{4(a^2+b^2)}+\frac{|a| x}{2(a^2+b^2)}+\frac{1}{4}\right).$$
With probability $1/4$ the scheme maps data and query points to 0; with probability  $1/2$, the scheme consists of anti bit-sampling with bias 0 and scaling factor $|a|/(a^2+b^2)\leq 1$; with probability $1/4$ the scheme consists of two anti bit-sampling with bias 0 and scaling factor $\sqrt{a^2+b^2}$ each.
We have $t^2 - 2a t + a^2 + b^2=4 (a^2+b^2) S_6(t,z)$.

\item \textbf{$(z,z')$ are conjugate complex roots, $-1\leq \textnormal{Real}(z)\leq  0$, and $|z|=a^2+b^2 < 1$.}
We use the  scheme of the previous point  with different parameters, giving CPF $$S_7(t,z)=\left(\frac{x^2}{4}+\frac{|a| x}{2}+\frac{a^2+b^2}{4}\right).$$
The scheme is the following:
with probability $1/4$, the scheme is a standard hashing scheme where data points are always mapped to $0$ and where a query point is mapped to $0$ with probability $a^2+b^2$ and to $1$ with probability $1-a^2+b^2$;
with probability  $1/2$, the scheme consists of anti bit-sampling with bias 0 and scaling factor $|a|\leq 1$; with probability $1/4$, the scheme consists of two anti bit-sampling with bias 0 and scaling factor $1$ each.
We have $t^2 - 2a t + a^2 + b^2=4  S_7(t,z)$.
\end{itemize}

Consider the scheme obtained by concatenating the above ones for each real root and each pair of conjugate roots. 
Its CPF is $S(t)=\prod_{i=1}^{6}\prod_{z\in Z_i} S_i(t,z)$, where $Z_i$ contains root with CPF $S_i$.
Then,  by letting $\psi$ denote the number of roots with negative real part, we get from Equation~\ref{eq:pol_fact}:
\begin{equation*}
\pp{t}= \left(2^{\psi} |a_k| \prod_{z\in Z, |\textnormal{Real}(z)|>1} |z|\right)  S(t) = \Delta S(t).
\end{equation*}

Consider now $a_k=0$ and let $\ell$ be the largest value such that $\pp{t}=t^\ell \ppp{x}$ with $\ppp{0}\neq 0$.
We get the claimed result by concatenating  $\ell$ anti bit-sampling, which gives a CPF of $x^\ell$, and the scheme for $\ppp{t}$ obtained by the procedure described above. 
\end{proof}